\documentclass[11pt,authoryear]{elsarticle}

\usepackage[left=2.cm, right=2.5cm, top=3.5cm, bottom=3cm, footskip=0.5cm]{geometry}
\usepackage{amsmath, amsfonts, amsthm, amssymb}
\usepackage{verbatim}
\usepackage{hyperref}
\usepackage{color}
\usepackage{xcolor}      
\usepackage{todonotes}   
\usepackage{graphicx}
\usepackage{cleveref}
\usepackage{enumerate}
\usepackage{wrapfig}
\usepackage{tikz}
\usepackage{empheq}
\usepackage{eurosym}
\usepackage{subcaption}
\usetikzlibrary{positioning}
\usepackage{stmaryrd}
\usepackage{algorithm}
\usepackage{algpseudocode}
\usepackage{ulem}
\usepackage{makecell}
\usepackage{cases} 
\usepackage{fancyhdr}
\usepackage{multirow}
\usepackage{mathrsfs} 

\usetikzlibrary{arrows.meta, positioning, shapes.geometric,
                calc, decorations.pathreplacing}

\numberwithin{equation}{section}
\numberwithin{equation}{section}
\newtheorem{theorem}{Theorem}[section]

\newtheorem{proposition}[theorem]{Proposition}

\theoremstyle{definition}
\newtheorem{definition}[theorem]{Definition}
\newtheorem{remark}[theorem]{Remark}
\newtheorem{assumption}[theorem]{Assumption}

\newtheorem{example}[theorem]{Example}
\numberwithin{equation}{section}

\newcommand{\RR}{\mathbb{R}}
\newcommand{\PP}{\mathbb{P}}
\newcommand{\FF}{\mathbb{F}}
\newcommand{\GG}{\mathbb{G}}

\newcommand{\NN}{\mathbb{N}}
\newcommand{\Ff}{\mathcal{F}}

\newcommand{\Jj}{\mathcal{I}}

\newcommand{\cW}{\mathcal{W}}

\newcommand{\cF}{\mathcal{F}}

\newcommand{\cG}{\mathcal{G}}

\newcommand{\cA}{\mathcal{A}}

\newcommand{\cR}{\mathcal{R}}

\newcommand{\EE}{\mathbb{E}}
\newcommand{\dr}{\mathrm{d}}
\newcommand{\OneN}{\{1,\hdots,H\}}

\newcommand{\VV}{\mathbb{V}}
\newcommand{\bOne}{\mathbf{1}}

\newcommand{\af}{\mathfrak{a}}

\newcommand{\ff}{\mathbf{F}}
\newcommand{\cc}{\mathfrak{c}}
\newcommand{\CC}{\mathfrak{C}}

\newcommand{\Jf}{\mathfrak{J}}

\newcommand{\aep}{\mathrm{AEP}}

\newcommand{\titre}{A stochastic SIR model for cyber contagion: application to granular growth of firms and to insurance portfolio}

\usepackage{fancyhdr}
\begin{document}

\begin{frontmatter}

\title{\titre}
\date{\today}

\author[1]{Caroline Hillairet}
\author[1]{Olivier Lopez}
\author[1,2]{Lionel Sopgoui}

\address[1]{CREST, UMR CNRS 9194, Ensae Paris, Avenue Henry Le Chatelier, 91120 Palaiseau, France}
\address[2]{Institut Louis Bachelier}

\journal{Scandinavian Actuarial Journal}

\begin{abstract}
\small
This work evaluates the impact of contagious cyber-events, over a finite horizon, on firms’ financial health and on a cyber insurance portfolio.
Our approach builds on key empirical findings from economics and cybersecurity. In economics, firm size and growth-rate distributions are non-Gaussian and exhibit heavy tails. In cybersecurity, contagion dynamics strongly depend on firm size and environmental conditions. 
To capture these features, we propose a stochastic multi-group SIR model coupled with a granular model of firm growth. This framework allows us to quantify the financial impact of cyber-attacks on firms’ revenues and on the insurer’s portfolio.
In the model, the arrival time and duration of cyber-attacks are driven by a combination of a Cox process and a Bernoulli random variable. The Cox process represents external contagion, with an intensity given by the force of infection derived from the stochastic SIR dynamics. The Bernoulli component captures contagion originating from an infected sister or subsidiary firm.
Environmental variability enables stochastic scenario generation and the computation of aggregate exceedance probabilities, a standard metric in catastrophe modeling that provides insurers with immediate insight into the financial severity of an event. We apply the framework to the LockBit ransomware attacks observed between May and July 2024. For a portfolio of 2,929 firms located in Île-de-France, the model predicts that, with 50\% probability, the insurer will need to compensate losses equivalent to up to two days of revenue over a 100-day cyber incident.
\end{abstract}

\end{frontmatter}

\footnotesize This research is part of the CyFi (Cyber Financialization) project, in cooperation with the cyber risk quantification tech startup \textit{Citalid} and with the financial support of \textit{Bpifrance}.


\normalsize
\section*{\label{sec:level1}Introduction}

Cybersecurity risk, alongside climate change and geopolitical instability, ranks among the most pressing emerging concerns for experts and the general population (see \cite{axa2024risks}). It encompasses malicious or accidental events that compromise the confidentiality, availability, or integrity of data and IT services. Common examples include malware infections, service-provider outages, and data breaches. Among these threats, ransomware -- a category of malware that
prevents or restricts users from accessing their system or exfiltrates sensitive data -- with lower barriers to entry, has become one of the most significant
threats. \cite{guidepoint2025} reports explosive growth in ransomware attacks: 76.8\% between 2022 and 2023, 8.72\% between 2023 and 2024, and 58\% between 2024 and 2025 in ransomware attacks. In 2025, this amounts to 7,515 publicly reported attacks attributed to up to 124 tracked ransomware groups. Victim organizations face substantial costs, including extortion payments --with \cite{axa2024risks} reporting \$1.1 billion paid to hackers in 2023--, remediation expenses, legal fees, business interruption losses, and reputational damage. That could consequently represent an increasing part in cyber insurance claims. As a result, there is a growing -- not only for businesses but also for financial institutions such as insurance companies, and even for governments -- to assess the economic and financial impact of potential future attacks. In this work, we focus on evaluating the impact of a cyber-episode on firms' financial health and on a cyber insurance portfolio.

Cyber risk modeling in finance and insurance, although relatively recent, encompasses a wide range of approaches. Early contributions rely on traditional frequency–severity frameworks \cite{eling2017data,farkas2021cyber}, while \cite{edwards2016hype} analyze trends in data breaches using Bayesian generalized linear models. More recently, Hawkes processes have been introduced to capture self-excitation and dependence in cyber events. In particular, \cite{bessy2021multivariate} and \cite{boumezoued2025cyber} propose multivariate Hawkes models to account for interdependencies between data breaches and the role of vulnerabilities in predicting attack frequency. Beyond point-process approaches, cyber-attacks share key mechanisms with biological epidemics, including the presence of a susceptible population, infection and recovery dynamics, peak intensity, prevalence, and network or cluster structures. This observation has motivated a growing literature inspired by compartmental epidemiological models. In this vein, \cite{hillairet2021propagation} adapt an epidemiological framework to model a contagious cyber event, such as a WannaCry-type attack, and apply it to an insurance portfolio to assess the impact of reaction and remediation measures. This approach is extended in \cite{hillairet2022cyber} to incorporate the network structure of the firm population. The present work builds on these contributions by designing a cyber episode based on the SIR model, a widely used epidemiological framework, as in  \cite{doenges2024sir,ji2011multigroup,bouzalmat2023stochastic}. In its simplest form, the SIR model partitions the population into three compartments: 
\begin{itemize}
    \item susceptible (S) firms that are exposed but not yet infected, 
    \item infected (I) firms that are affected and can transmit the malware, \item and removed (R) firms that have recovered or been secured.
\end{itemize}
The dynamics of the epidemic are governed by a system of ordinary differential equations describing the transitions between these compartments.

As the application of epidemiological models to cyber insurance remains at an early stage, several important features have yet to be fully incorporated. In this work, we focus on integrating key empirical stylized facts from both economics and cybersecurity. First, firm size and firm structure exhibit pronounced heterogeneity. In particular, \cite{axtell2001zipf} shows that firm size distributions are heavy-tailed, often following Zipf’s law, while \cite{stanley1996scaling} documents that firm growth-rate distributions are non-Gaussian. Second, empirical evidence in cybersecurity indicates that the dynamics of cyber-episodes depend strongly on firm size: larger firms are more frequently targeted by cyber-attacks, see \cite{BaksyCyberFirmSize2025,kamiya2021risk}. Moreover, contagion parameters are subject to environmental variability as pointed out by \cite{ECBMacroprudentialBulletin2025}. For instance, in firms organized into multiple subunits, infection rates may differ substantially depending on whether contamination originates from within the organization or from external sources. Accounting for these sources of heterogeneity is essential for realistically modeling cyber contagion and its financial consequences.

To account for these features, we propose a framework coupling a stochastic multigroup SIR model with a granular growth model for firm revenue. Following the epidemiological approach of \cite{doenges2024sir}, we study the influence of firm size and distribution on attack propagation. However, rather than the standard conversion of deterministic ODEs into Stochastic Differential Equations (SDEs), we adopt a more intuitive approach by embedding stochasticity directly into the parameters via Cox-Ingersoll-Ross (CIR) processes, as in \cite{allen2016environmental}. To model revenue, we adapt the granular growth framework of \cite{moran2024revisiting}, which treats each firm as a collection of subunits with varying degrees of independence. By extending this model into continuous time, we can capture the instantaneous revenue loss occurring at the subunit level during an active attack. The arrival of these attacks is governed by a hybrid process:
\begin{itemize}
    \item External Contagion: a Cox process with an intensity -- the "force of infection"-- derived from the SIR model dynamics.
    \item Internal Contagion: a Bernoulli random variable representing secondary infection from a sister subsidiary.
\end{itemize}
This mixed process defines both the initial arrival time and the duration of the contagion within a subunit. Finally, we define the insurer's instantaneous cost as the revenue differential between the "no-attack" and "attack" scenarios. At the portfolio level, we quantify risk using the Aggregate Exceedance Probability (AEP) to determine the probability that total losses exceed a defined threshold over a specific horizon.
Throughout the paper, 
we consider a cyber-episode over a finite horizon.

The remainder of this paper is organized as follows. In Section~\ref{sec:the granular model}, we introduce the jump-diffusion model used to determine the financial health of firms subject to cyber-attacks; we also define the Aggregate Exceedance Probability (AEP) as our primary risk measure for the insurance portfolio. Section~\ref{sec:the contagion model} is dedicated to the cyber-attack contagion model, providing analytical results that establish the existence and uniqueness of a non-negative solution for the stochastic multi-group SIR framework. In Section~\ref{sec:linked}, we bridge these frameworks by translating the contagious cyber-episode into the jump-diffusion revenue model using internal and external transmission probabilities. Finally, Section~\ref{sec:simulations} details our data sources, model calibration, and simulation results. Specifically, the contagion model is validated using data from the ransomware episode attributed to the Lockbit group (May - July 2024), and its impact is assessed on a representative portfolio of 2,929 firms in the Île-de-France region.

\section{Firm value and insurance portfolio dynamics under cyber-attacks}\label{sec:the granular model}

In this section, we extend the granular model of firm growth of \cite{wyart2003statistical,herskovic2020firm,moran2024revisiting} in a continuous time setting and in the context of cyber risk. This model has been introduced to fit some empirical facts, in particular the ones cited by \cite{stanley1996scaling,axa2024risks}: both the firm size
distribution and the growth rate distribution are non-Gaussian and
feature heavy tails.

Throughout this work we consider a probability space \((\Omega,\cA,\PP)\), where $\cA$ is a $\sigma$-algebra on $\Omega$. We fix $T > 0$ which represents the horizon of the cyber-episode. 
In the same spirit as \cite{moran2024revisiting}, within a given insurance portfolio of $H \in\NN^*$ firms characterized by exogenous production, we consider that firms are composed of a number of independent subunits, such as departments or production units, which operate within separate sub-markets and which are not necessarily of equal size. We have the following assumption.
\begin{assumption}\label{ass:firm revenue}
    The revenue of firm $i\in\OneN$ at time $t\geq 0$ is denoted by $Z_{i,t}$ and satisfies 
\begin{align}\label{eq:firm revenue}
    Z_{i,t} := \sum_{j=1}^{K_i} z_{ij,t},
\end{align}
where $K_i\in\NN^*$ represents
the number of subunits and $z_{ij,t}$, $j = 1,\hdots, K_i$, denotes the respective revenue.
\end{assumption}
\noindent We will note
\begin{align}
    \Jj := \left\{(i,j) \mid i\in\{1,\hdots,H\}\text{ and }j\in\{1,\hdots,K_i\} \right\}.
\end{align}
The firm's size is the number of its subunits, and the revenue is an indicator of firms’ and subunits' financial health. Instead of revenue, we could consider its  workforce, profits, etc. In addition, we could view the economy as a supra-firm subdivided in sectors. In the same way, each sector can be seen as a supra-firm subdivided in companies. This point of view could be interesting for calibration because public data are usually more easily available at economy/sectoral level than at firm's level.
The firm's size $K_i$, which is an integer, is a random variable assumed to be distributed according to Zipf's law (the discrete version of Pareto distribution) on $\{1,\hdots, K\}$, $K\in\NN^*$. For $k=1, \cdots, K$
\begin{equation}\label{eq:pareto sub-unit}
    \PP(K_i=k) = q\frac{\af k^{-(1+\af)}}{1-K^{-\af}};\quad 1 < \af < 2.
\end{equation}
The above assumption is motivated by the literature (see \cite{wyart2003statistical, moran2024revisiting}) and is easily verified on the data (see \Cref{fig:histogram_sub_unit} and \Cref{fig:zipf_number_sub_unit} where $\af = 1.76$ and $q = 0.784$).

\subsection{The impact of cyber-attacks on firms’ revenue}

When a cyber-attack succeeds, it causes financial loss to the company. This loss can result from paying a ransom, business disruption, repairing equipment, reputation costs, etc. Let $1\leq i\leq H$ and $1\leq j\leq K_i$. In the absence of a cyber-attack, the
revenue $z_{ij}$  of subunit $j$ of firm $i$ is modeled as a Geometric Brownian motion,  while in the presence of a cyber-attack, it is modeled as a jump-diffusive  process. We introduce the following assumption.
\begin{assumption}\label{ass:sub-unit dynamics}
    For $(i,j)\in\Jj$,
    \begin{enumerate}
        \item Without the cyber-event, the time evolution of each subunit $i$’s revenue, $\overline z_{ij}$, is characterized by the diffusion process
        \begin{align}\label{eq:sub-unit dynamics wo cyber}
    \frac{d\overline z_{ij,t}}{\overline z_{ij,t}} = \mu_{ij} \dr t + \sigma_{ij} \dr  B_{ij,t},\qquad \overline z_{ij,0} = z_{ij,0}> 0, \; \sigma_{ij} > 0, \; \mu_{ij}\in\RR,
\end{align}
where $B_{ij}$ is a Brownian motion. For a given firm $i$, $1\leq i\leq H$, the Brownian motions $(B_{ij})_{1\leq j\leq K_i}$ driving the revenues of the subunits are  correlated with correlation coefficient $\rho_i\in\RR$. But  for different firms  $i_1\neq i_2$,  
$B_{i_1 j_1}$ and $B_{i_2 j_2}$ are assumed to be independent. Integrating SDE \eqref{eq:sub-unit dynamics wo cyber}, we have
\begin{align*}
    \overline z_{ij,t} = z_{ij,0}  \exp{\left(\left(\mu_{ij} -\frac{\sigma_{ij}^2}{2} \right)t + \sigma_{ij} B_{ij,t}\right)}.
\end{align*}
    \item The random time $\tau_{ij}$ of the arrival  of the cyber-attack affecting the subunit $j$ of firm $i$,  is defined as 
    \begin{align}\label{eq:first arrival time}
        \tau_{ij} = \inf\{t\geq 0|N_{ij,t} \geq 1  \},
    \end{align}
    where $N_{ij,t}$ is a point process that will be defined later.
    \item $\delta_{ij}$ is a random time representing the duration of an infection.
    \item We introduce the process $\pi_{ij}$, taking values in $[0,1]$, which represents the random fraction of the firm’s revenue lost during the infection. 
    \end{enumerate}
    Therefore, the firm's revenue $z_{ij,t}$ evolves as follows
\begin{align}\label{eq:sub-unit dynamics}
    z_{ij,t} =  \left\{ 
    \begin{array}{l}
      \overline z_{ij,t} \qquad\qquad\qquad\text{ if } t < \tau_{ij}\\
      \left(1-\pi_{ij,t} \right) \overline z_{ij,t} \qquad\text{if } \tau_{ij} \leq t < \tau_{ij} + \delta_{ij}\\
        \overline z_{ij,t} \qquad\qquad\qquad\text{ if } t \geq  \tau_{ij} + \delta_{ij}
    \end{array}
  \right.
\end{align}
Moreover, $(\pi_{ij})_{(i,j)\in\Jj}$, $(N_{ij})_{(i,j)\in\Jj}$, and $(B_{ij})_{(i,j)\in\Jj}$ are assumed to be independent.
\end{assumption}
Let $(i,j)\in\Jj$. In the above assumption, the constant parameters $\sigma_{ij}$ (resp. $\mu_{ij}$) define the order of magnitude (resp. the drift) of  the growth fluctuations at the level of a subunit. Regarding random time $\tau_{ij}$, for $t\geq 0$ and the subunit $(i,j)$,  the event $\{\tau_{ij}\leq t\}$
corresponds to a successful breach caused by an attack on  before time $t$; otherwise, no attack has taken place yet or the potential attack is blocked and no loss occurs before time $t$. It should also be noted that the $K_i$-dimensional process $(N_{ij})_{1\leq j\leq K_i}$ is potentially correlated, because there can be internal contamination between subsidiaries, as we will see later.
Moreover, the property on  the Brownian motion $(B_{ij})_{ij}$ means that inside the same firm, the subunits' revenues are correlated but different firms' revenues are not.

Finally, the process $\pi_{ij}$ represents the severity (i.e. the instantaneous random monetary loss) of the attack. This loss may be potentially covered by an insurance contract.  It is a $[0,1]$-value process because a cyber-attack is detrimental to the subunit, but its losses remain capped at its potential earnings. The time-behavior of the loss  process may depend on the dynamics of the epidemic, on the firm’s investments in cybersecurity, and also on cyber insurance premium: the severity might begin high when infection starts, and progressively drop to zero (full recovery). Or, it could start slowly, peak, and then decline.  
Having no explicit empirical input on this feature, we limit ourselves in this work to the case where the severity does not depend on time, so $\pi_{ij}$ is a random variable instead of a stochastic process. The sequence $(\pi_{ij})_{ij}$ are assumed  independent and identically distributed with common distribution $f_\pi$ having  support in $[0,1]$, mean $\pi_\star$ and standard deviation $\sigma_\star^2$.

\begin{example}[Beta distribution]  
For the numerical analysis,  we will take $\pi_{ij} \sim B(\alpha^\pi, \beta^\pi)$ whose density probability function is $f(x;\alpha^\pi, \beta^\pi) = \frac{x^{\alpha^\pi-1}(1-x)^{\beta^\pi-1}}{B(\alpha^\pi, \beta^\pi)}$ for $x\in[0,1]$ where $B(\alpha^\pi, \beta^\pi) = \frac{\Gamma(\alpha^\pi) \Gamma(\beta^\pi)}{\Gamma(\alpha^\pi+\beta^\pi)}$ and $\Gamma$ is the Gamma function.
\end{example}

 If we are dealing with the whole portfolio (or a whole sector or an entire economy), we can also derive the total  output (revenue) $\mathbf{O}$ so that for all $t\geq 0$,
 \begin{align}\label{eq:whole economy}
     \mathbf{O}_t := \sum_{i=1}^{H} Z_{i,t},
 \end{align}
 where the individual firm revenue is given by 
\begin{align*}
    Z_{i,t} = \sum_{j=1}^{K_i} \left(1-\pi_{ij} \bOne_{ \tau_{ij} \leq t < \tau_{ij} + \delta_{ij}} \right)\overline z_{ij,t}= \sum_{j=1}^{K_i}   \left(1-\pi_{ij} \bOne_{ \tau_{ij} \leq t < \tau_{ij} + \delta_{ij}} \right) z_{ij,0} \exp{\left(\left(\mu_{ij} -\frac{\sigma_{ij}^2 }{2}\right)t + \sigma_{ij} B_{ij,t}\right)}.
\end{align*}
We also introduce the natural filtration of $(B_{ij})_{(i,j)\in\Jj}$
\begin{align}\label{eq:filtration Fb}
    \FF^B := (\cF^B_t)_{t\geq 0}.
\end{align}

\subsection{The impact of cyber-attack on insurance portfolio}

Based on the description of the attack at the policyholder level, the insurance company is
interested in aggregating these risks. We focus on the total claims of the cyber-event for the insurance company over a period. For a policyholder experiencing a cyber-attack, each day spent infected results in financial losses: their revenue decreases. The insurance company potentially compensates a part (or all) of these losses. Consider the policyholder $i\in\OneN$. For all $j\in\{1,\hdots,K_i\}$, recall the dynamics of the subunit's revenue without and with the cyber risk $\overline z_{ij}$ and $z_{ij}$ respectively in \eqref{eq:sub-unit dynamics wo cyber} and in \eqref{eq:sub-unit dynamics}.
We  then define the instantaneous claim for policyholder $i$: 
\begin{definition}
     For $1\leq i\leq H$, the process $\cc_{i}$, corresponding to the instantaneous claim for policyholder $i$ at time $t\geq 0$, is defined as
     \begin{align}\label{eq:explicit c_i}
         c_{i,t}: = -\sum_{j=1}^{K_i} \left(z_{ij,t} - \overline z_{ij,t}\right) =\sum_{j=1}^{K_i} \pi_{ij} \bOne_{ \tau_{ij} \leq t < \tau_{ij} + \delta_{ij}}z_{ij,0} \exp{\left(\left(\mu_{ij} -\frac{\sigma_{ij}^2}{2} \right)t + \sigma_{ij} B_{ij,t}\right)}.
     \end{align}
\end{definition}
\noindent This represents the sum of the differences in income for each sub-unit, with and without a cyber-attack.\\ 

A key metric for the insurer for quantifying the impact of attacks is the Cumulative Distribution Function (CDF) of the portfolio’s total losses, 
with an emphasis on extreme losses  that could threaten liquidity.
We focus first on losses over short intervals, such as a single day.


\begin{definition}
For $t,u\geq 0$, the cumulative distribution function of the portfolio  over  the period $[t,t+u]$ is denoted, for all $x\geq 0$,
\begin{align}\label{eq:cdf portfolio}
    \ff_{[t,t+u]}(x) &:= \PP\left(\sum_{i=1}^{H}C_{i,[t,t+u]} \leq x\right),
\end{align}
 where for  $1\leq i\leq H$, $C_{i,[t,t+u]}$ is  the total claim of firm $i$ between $t$ and $t+ u$
\begin{equation}\label{eq:claim in period}
\begin{split}
    C_{i,[t,t+u]} := \int_{t}^{t+u} c_{i,s} \dr s= \sum_{j=1}^{K_i} z_{ij,0} \int_{t\vee\tau_{ij}}^{(t+u)\wedge(\tau_{ij} + \delta_{ij})} \pi_{ij} e^{\left(\mu_{ij} -\frac{\sigma_{ij}^2}{2} \right)s + \sigma_{ij} B_{ij,s}}\dr s.
\end{split}
\end{equation}
\end{definition}
In order to measure the overall severity  of the attack, we next introduce a widely used metric in catastrophe risk modeling: the Aggregate Exceedance Probability or AEP (see \cite{grossi2005catastrophe,homer2017notes}). 
It is the probability that the sum of {several cyber-incidents}
losses  over a period exceeds a certain amount. In general, AEP is computed on one year, but we define here AEP on the horizon of the epidemic $T$.
For a {given} cyber-episode lasting until date $T$, the total portfolio loss due to {this} episode over $[0,T]$, is denoted $\CC_T$
\begin{align}\label{eq:c portfolio}
    \CC_T := \sum_{i=1}^{H}\left( \int_{0}^{T}c_{i,t}\dr t\right) =  \sum_{i=1}^{H} C_{i,[0,T]}.
\end{align}
The AEP calculates the impact of several identical cyber-incidents that occurred during the period $[0,T]$.

\begin{definition} Over the period $[0,T]$, the insurance portfolio can be affected by  a random number $P$ 
of cyber-events, each of them leading to a random claim amount $\CC_T^p$, $p=1,\cdots, P$. We assume the random claim amounts $(\CC_T^p)_p$ to be  independent and identically distributed according to $\ff_{[0,T]}$ (that is the distribution of $\CC_T$ defined in \eqref{eq:c portfolio}), and independent of $P$.
Then the Aggregate Exceedance Probability (AEP)  at time $T$ is defined as, for $x\geq 0$,
\begin{align}\label{def:portfolio measure}
    \aep_T(x) &:= \PP\left(\CC_T^1 + \CC_T^2 + \hdots + \CC_T^P > x\right).
\end{align}
\end{definition}
\medskip
Computing $\aep_T$  requires the convolution of  the CDF of $\CC_T$, $\ff_{[0,T]}$ defined in \eqref{eq:cdf portfolio}.
For $p \in \NN$, let
 $\ff_{\CC_T}^{(p)}$ denote the $p$-fold convolution of $\ff_{[0,T]}$ defined iteratively as
\begin{align*}
 \ff_{\CC_T}^{(1)} = \ff_{[0,T]}, \quad    \ff_{\CC_T}^{(p)}(x) = \int_{0}^{x} \ff_{\CC_T}^{(p-1)}(x-y) \mathbf{f}_{[0,T]}(y)\dr y, \quad p = 2, 3, \hdots,
\end{align*}
with  $\mathbf{f}_{[0,T]}$ the probability density function associated with $\ff_{[0,T]}$.
Then, 
given that $P$ is  independent of the claim amounts $(\CC^p_T)$, 
we have by the law of total probability, for\footnote{Obviously $\aep_T(0)=1$.} $x>0$
\begin{equation}\label{eq:portfolio measure}
\begin{split}
    \aep_T(x) &= \sum_{p=1}^{+\infty} \PP\left(\CC_T^1 + \CC_T^2 + \hdots + \CC_T^p > x\middle|P = p\right) \PP[P = p]\\
    &= 1 - \sum_{p=1}^{+\infty}  \ff_{\CC_T}^{(p)}(x)\PP[P = p]= 1 - \EE_P\left[\ff_{\CC_T}^{(P)}(x) \right]\\
    &= \EE_P\left[\PP\left(\sum_{k=1}^{P}\CC_T^k > x\right)\right].
\end{split}
\end{equation}
 In the numerical section, the latter quantity is computed  via a Monte Carlo simulation.

All the above metrics require to know the distribution  of the processes $(c_{i,t})_{t}$ for each $i$. The sources of randomness in these quantities are the processes $(B_{ij,t})$, $(N_{ij,t})$, and the variables $(\pi_{ij})$. In order to evaluate the impact of cyber incidents on an economy (using \eqref{eq:whole economy}) or on an insurance portfolio (using \eqref{eq:portfolio measure}), a crucial point is to model the cyber contagion or in other words,  the dynamics of the process $(N_{ij})_{ij}$.  Early works such as \cite{peng2017modeling} or \cite{zeller2022comprehensive} model the arrival of cyber-attacks using (in)homogeneous Poisson processes. Thereafter, \cite{bessy2021multivariate} and \cite{boumezoued2025cyber} show that Hawkes processes are particularly
suitable to capture the contagion phenomena, the clustering, and the autocorrelation of the arrival times of cyber incidents. Lastly, by noting that a cyber incident and a biological epidemic have similar characteristics, \cite{hillairet2021propagation} and \cite{hillairet2022cyber} use compartmental epidemiological models to describe cyber contagion by linking the hazard rate function of arrival times with the dynamics of infected individuals. We adapt the latter approach in the present work because our purpose is to model the dynamics of a specific cyber-contagion over a finite time horizon.

\section{The contagion model} \label{sec:the contagion model}
As suggested by \Cref{fig:GRIT_ransomware} (from the 2025 Ransomware \& Cyber Threat Report of \textit{GuidePoint Security’s Research and Intelligence Team's [GRIT]}),  biological epidemics and cyber contagion episodes have fairly similar characteristics. Much like certain contagious biological outbreaks, a firm (just as a person) could be susceptible i.e. not yet infected but vulnerable if exposed to an infectious firm, infected i.e. affected by the malware and which can transmit, and removed or recovered i.e.  whose systems are cured, secured, or patched. Moreover, when examining the dynamics of companies infected by the LockBit ransomware between May and July 2024 on \Cref{fig:Number_infected_LockBit_2024}, the number of infections starts out low, then escalates rapidly until it reaches a peak, before beginning to decline. Finally, just as infection rates differ within and between households in biological epidemics, cyber epidemics mirror this transmission pattern within and between  firms’ subunits. Subsidiaries of a given company are, by nature, more interconnected together than with external organizations.

\begin{figure}[ht!]
  \centering
  \begin{subfigure}[b]{0.59\textwidth}
    \centering
    \includegraphics[width=\textwidth]{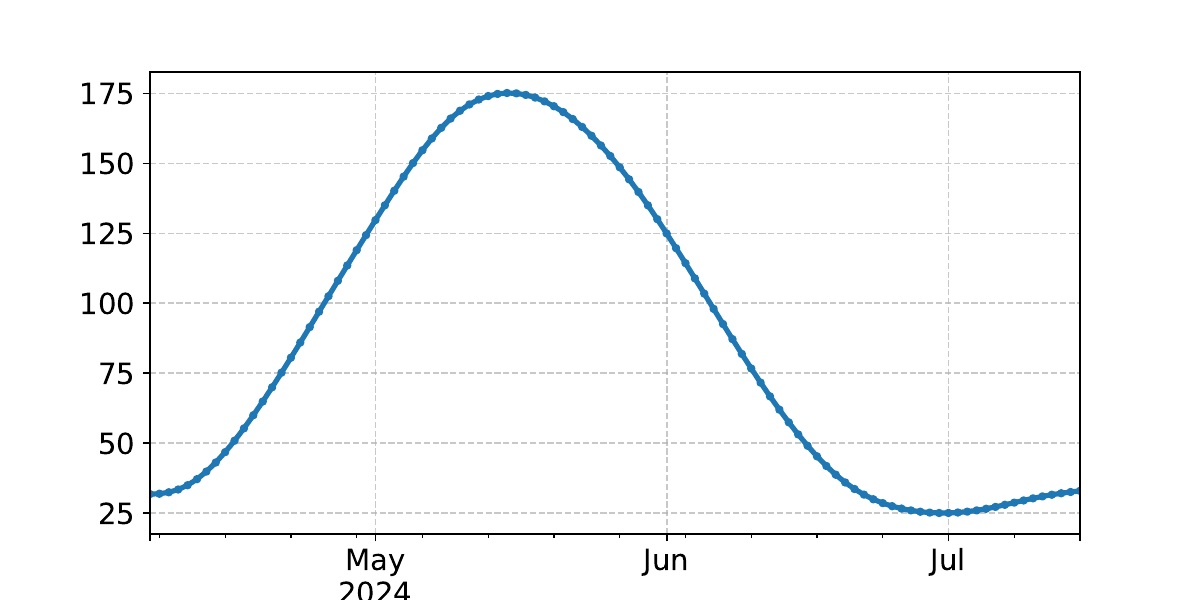}
    \caption{Number of infected from LockBit reported from May to July 2024 (Date in days on the x-axis and number of companies attacked on the y-axis)}
    \label{fig:Number_infected_LockBit_2024}
  \end{subfigure}
  \hfill
  \begin{subfigure}[b]{0.4\textwidth}
    \centering
    \includegraphics[width=\textwidth]{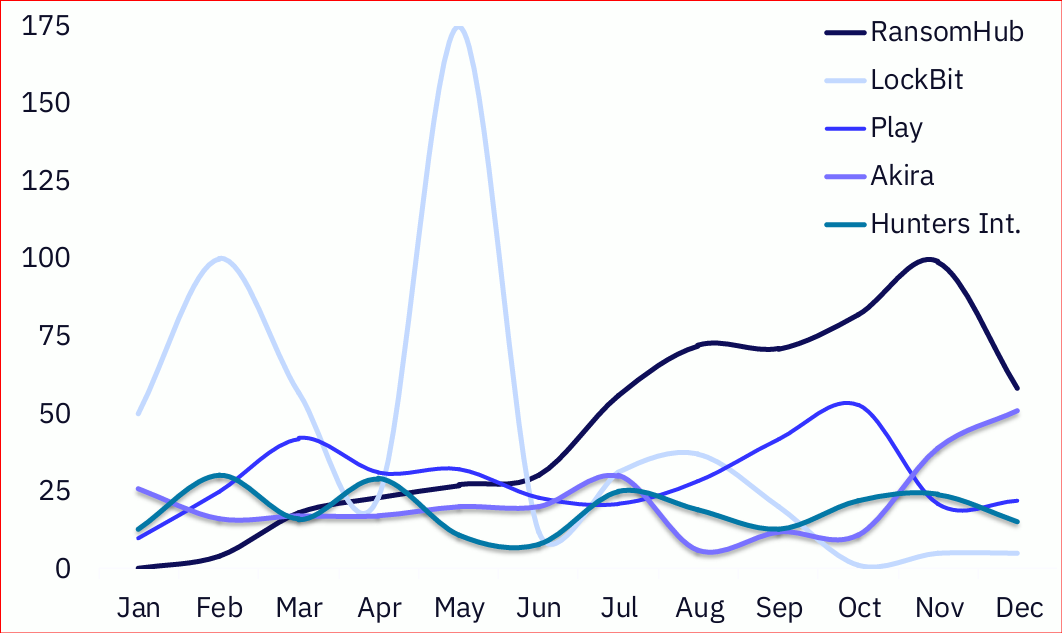}
    \caption{Most Impactful Ransomware in 2024 (the date in months on the x-axis and the number of companies attacked on the y-axis)}
    \label{fig:Most_Impactful_Ransomware_2024}
  \end{subfigure}
  \caption{GRIT 2025 Ransomware \& Cyber Threat Report (see \cite{guidepoint2025})}
  \label{fig:GRIT_ransomware}
\end{figure}

\subsection{The model}
To model the infection process within a firm, we adapt the SIR for households proposed by \cite{doenges2024sir}. Precisely, we extend their deterministic SIR into a stochastic SIR in the sense that the model parameters are no longer deterministic but instead follow a Cox–Ingersoll–Ross dynamic. As \cite{hillairet2021propagation}, we assume that contagion does not come from inside the portfolio itself, but from  outside, so that the spread of the attack is defined as a global level. We
view a cyber infectious event in a fully susceptible firm of size $k\in\{1,\hdots,K\}$  
as a splitting process generating a infected firm of size $j$, where $1 \leq j \leq k$ and a remaining
susceptible firm of size $k-j$. Let $s_j$, $i_j$ and $r_j$ denote the number of
 susceptible, infected, and removed firms of size $j$, where $1 \leq j \leq K$. The quantity $K$, introduced in \Cref{eq:pareto sub-unit}, represents the maximal size of a firm. We denote the processes
\begin{itemize}
    \item $h_j = s_j + i_j + r_j$ equals to the total number of current firms of
size $j$;
\item $h = \sum_{j=1}^{K} h_j$ is the total number of firms;
\item The total population (number of subunits) is given by $n = \sum_{j=1}^{K} j h_j$.
\end{itemize}
Rather than working with the number of susceptible/infected/removed firms which are integers and not fully adapted to the ODE, we introduce for all $1\leq j\leq K$, $S_j := s_j/h$, $I_j := i_j/h$, $R_j := r_j/h$, and $H_j := h_j/h$ the fraction of susceptible, infected, removed, and all firms of size $j$ respectively. We thus write the average size of firm as
\begin{align}\label{eq:total population}
    N = \sum_{j=1}^{K} j H_j.
\end{align}
Let us now describe how the cyber-infection spreads into the subunits. Let 
$j$ range from $1$ to $K$.
\begin{enumerate}
    \item If an initial infection is brought into a susceptible firm of size $j$, secondary infections will occur inside the firm.
    \item Each of the remaining $j-1$ firm members can get infected with equal probability $a$. This represents the "in-firm" infection rate.
    \item The probability that a primary infection in a firm of size $j$ generates in total $k$ infections inside this firm is $b_{j,k}$, where $1 \leq k \leq j$. 
    \item The secondary infections give rise to a splitting of the initial firm of size $j$ into a new, fully infected firm of size $k$ and another still susceptible firm of size $j-k$.
    \item An infected firm of size $k$ contributes to the overall force of infection between different firms and recovers with a rate $\gamma_k$. 
    \item $\beta_k$ is the ”out–firm” infection rate which refers to infection events occurring between different firms and hence outside a given single firm.  
\end{enumerate}
\begin{figure}[ht]
\centering
\scalebox{0.75}{
\begin{tikzpicture}[
    >=Stealth,
    node distance = 1.0cm,
    sus/.style={circle, draw=blue!70!black, fill=blue!20, minimum size=6mm, thick},
    inf/.style={circle, draw=red!70!black, fill=red!40, minimum size=6mm, thick},
    rec/.style={circle, draw=green!50!black, fill=green!30, minimum size=6mm, thick},
    state_label/.style={font=\small\sffamily\bfseries, align=center},
    rate_label/.style={font=\scriptsize\sffamily, fill=white, inner sep=1pt},
    split_arr/.style={->, thick, shorten >=2pt, shorten <=2pt, red!70!black, ultra thick}
]

\node[state_label] (Slabel) at (0,0) {One initial Susceptible firm of $S_5$};
\foreach \i in {1,...,5} {
    \node[sus] (S\i) at (\i*0.8-2.4,-0.8) {};
}
\node[font=\tiny] at (S1) {1};
\node[font=\tiny] at (S3) {3};
\node[font=\tiny] at (S2) {2};
\node[font=\tiny] at (S4) {4};
\node[font=\tiny] at (S5) {5};

\draw[split_arr] (0,-1.5) -- (0,-2.2) 
    node[midway,right, rate_label] {$b_{5,2}$};

\node[state_label, align=left] (Ilabel) at (-3,-2.8) {One Infected firm of $I_2$};
\node[inf] (I1) at (-3.5,-3.6) {1};
\node[inf] (I2) at (-2.5,-3.6) {3};

\draw[->, dotted, thin, red!70!black] (S1.south) to[bend right=20] (I1.north);
\draw[->, dotted, thin, red!70!black] (S3.south) to[bend left=20] (I2.north);

\node[state_label, align=left] (S3label) at (3,-2.8) {One Susceptible firm of $S_3$};
\node[sus] (S3_2) at (2.2+0.8*1,-3.6) {2};
\node[sus] (S3_4) at (2.2+0.8*2,-3.6) {4};
\node[sus] (S3_5) at (2.2+0.8*3,-3.6) {5};

\draw[->, dotted, thin, blue!70!black] (S2.south) to[bend right=10] (S3_2.north);
\draw[->, dotted, thin, blue!70!black] (S4.south) -- (S3_4.north);
\draw[->, dotted, thin, blue!70!black] (S5.south) to[bend left=10] (S3_5.north);

\node[state_label, align=left] (Rlabel) at (-3,-4.8) {One Recovered firm of $R_2$};
\node[rec] (R1) at (-3.5,-5.6) {1};
\node[rec] (R2) at (-2.5,-5.6) {3};

\draw[->, thick, green!50!black] (I1.south) -- (R1.north) 
    node[midway, left, rate_label, xshift=-2mm] {$\gamma_2$};
\draw[->, thick, green!50!black] (I2.south) -- (R2.north);

\node[state_label] (Other) at (3,-5.2) {Other Susceptible firms of $S_j$};
\draw[->, thick, red!70!black] (I2.east) .. controls (0,-4.2) and (1.5,-5.2) .. (Other.west)
    node[midway,above, rate_label] {$\beta_2$};

\end{tikzpicture}
}
\caption{Infection and Splitting: $S_5 \to I_2 + S_3$}
\label{fig:splitting}
\end{figure}
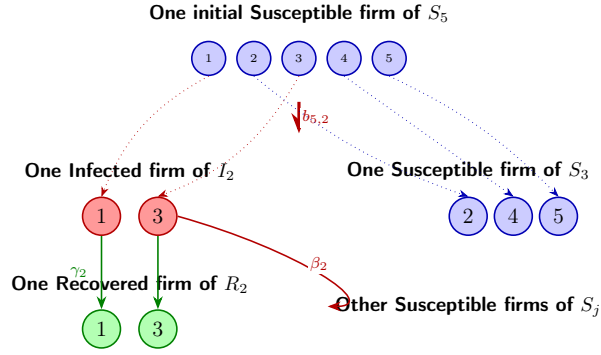 
The graph in \Cref{fig:splitting} above gives an example where $k=5$ and $j=2$. We summarize and precise these features in the following assumption.
\begin{assumption}\label{ass:SIR}
    The dynamic system governing the dynamics of the susceptible, infected
and removed firms of size $k$ reads as
\begin{empheq}[left=\empheqlbrace]{align}
    \frac{\dr S_{k,t}}{\dr t} &= Y_t \left(-k S_{k,t}+ \sum_{j=k+1}^{K} j S_{j,t}\cdot b_{j j-k, t}\right),\label{eq:susceptibles}\\
    \frac{\dr I_{k,t}}{\dr t} &= -\gamma_{k,t} I_{k,t} + Y_t \sum_{j=k}^{K} j S_{j,t}\cdot b_{jk,t} ,\label{eq:infected}\\
    \frac{\dr R_{k,t}}{\dr t} &= \gamma_{k,t} I_{k,t},\label{eq:recovered}
\end{empheq}
where
\begin{align}
    Y_t &= \frac{1}{N_t} \sum_{k=1}^{K} \beta_{k,t} \cdot k I_{k,t}\label{eq:force}
\end{align}
is the total force of infection -- the instantaneous risk that a susceptible individual becomes infected. Let us note the set of  $(2K + K^2)$ stochastic parameters
\begin{align*}
    \Psi &:= \left\{\beta_1,\hdots,\beta_K, \gamma_1,\hdots,\gamma_K, b_{1 1},\hdots,b_{1 K},b_{2 1},\hdots,b_{K-1 K},b_{K K}\right\}.
\end{align*}
We assume that a parameter  $\varphi\in\left\{b_{1 1},\hdots,b_{1 K},b_{2 1},\hdots,\allowbreak b_{K-1 K},b_{K K}\right\}$ is a $[0,1]$-valued process \\$\varphi = \frac{1}{1+e^{-\tilde{\varphi}}}$ where $\tilde\varphi$ evolves according to the Cox-Ingersoll-Ross (CIR) dynamics
\begin{align}\label{eq:CIR}
\dr\tilde\varphi_t &= \kappa_\varphi(\mu_\varphi-\tilde\varphi_t) \dr t + \Sigma_\varphi\sqrt{\tilde\varphi_t} \dr \cW_{\varphi,t}, \qquad \mu_\varphi\in \RR, \; \Sigma_\varphi>0, \; \kappa_\varphi >0, \; \varphi_0 \geq 0,
     \end{align}
     while $\varphi\in\left\{\beta_1,\hdots,\beta_K, \gamma_1,\hdots,\gamma_K\right\}$  is also assumed to be a CIR process.\\ 
     In both case, the parameters satisfy Feller conditions i.e. $2 \kappa_\varphi\mu_\varphi \geq \Sigma_\varphi^2$. Moreover, $(\cW_{\varphi})_{\varphi\in\Psi}$ is a $(2K+K^2)$-dimensional Brownian motion, independent to $(B_{ij})_{(i,j)\in\Jj}$ (defined in Assumption~\ref{ass:sub-unit dynamics}).
\end{assumption}
Recall that the expectation and variance of a CIR process \eqref{eq:CIR} is given by (see \cite{cox1985theory}): 
for  $t\geq 0$, 
\begin{align*}
    \EE[\tilde\varphi_t] = \varphi_0 e^{-\kappa_\varphi t} + \mu_\varphi (1-e^{-\kappa_\varphi t} ),
\quad 
    \VV[\tilde\varphi_t] = \varphi_0\frac{\Sigma_\varphi^2}{\kappa_\varphi}\left(e^{-\kappa_\varphi t} - e^{-2\kappa_\varphi t} \right) + \frac{\mu_\varphi \Sigma_\varphi^2}{2\kappa_\varphi}\left(1-e^{-\kappa_\varphi t}\right)^2.
\end{align*}

\begin{remark}\label{rem:total population}
    In this SIR model, the average size process defined in \eqref{eq:total population}
\begin{align*}
 N = \sum_{k=1}^{K} k (S_k + I_k + R_k),
\end{align*}
 does not depend on $t$. From now on, we will write $N_0$ instead of $N_t$ to refer to the average size.
\end{remark}

We could have used a SIRS model that allows an infected-and-removed unit to become susceptible again and be reinfected, but we can ignore this feature since we are dealing with a short-lived cyber incident. In other words, there is not enough time for a unit to be infected multiple times.\\
The details on the deterministic version of this model are detailed in \cite{doenges2024sir}. The stochasticity in the epidemiological model is motivated by the fact that  fluctuations in the environment affect models' parameters. Several works in the literature introduce randomness by simply transforming the model from a system of ODEs to a system of SDEs. This implies that the randomness comes from a single parameter (for example on transmission coefficient between compartments for \cite{ji2011multigroup} or on the the decay rate of immunity for \cite{bouzalmat2023stochastic}). 
However, since all the parameters can be subject to environmental variability, it appears more natural to consider each parameter having its own stochastic dynamics.

We define the (right continuous and complete) filtration $\FF^\cW = (\cF^\cW_t )_{t\geq 0}$  generated by the stochastic  parameters:  for $t\geq 0$,
\begin{align}\label{eq:filtration F}
    \Ff^\cW_t := \sigma\left(\left\{\cW_{\varphi,s}, \;   s\leq t\text{ and }\varphi\in\Psi\right\}\right).
\end{align}

\begin{remark}\label{rem:simplify b}
For the numerical analysis, the infections inside the firm is modeled by a Bernoulli-process with in-firm attack rate (infection probability) $a \in [0, 1]$, and the total
number of infected subunit inside the firm follows a binomial distribution, namely, 
for $1\leq j,k\leq K$, 
\begin{align}
    b_{jk} = \binom{j-1}{k-1} a^{k-1} (1 - a)^{j - k}.
\end{align}
Next, we also assume that  $a = \frac{1}{1+e^{-\tilde{a}}}\in[0,1]$ where the process $\tilde{a}$ is governed by a CIR dynamics. This assumption reduces the dimension of the set $\Psi$ from $2K+K^2$ to $2K + 1$  
and  ease the calibration process.
\end{remark}

\subsection{Existence and uniqueness of the nonnegative solution of the SIR system} 

This section is dedicated to the existence and uniqueness result of a  nonnegative solution to the system \eqref{eq:susceptibles}-\eqref{eq:infected}-\eqref{eq:recovered}. The following theorem and its proof are inspired by \cite{ji2011multigroup}.
\begin{theorem}[Existence and uniqueness]\label{theorem:existence and unicity}
    Consider the system \eqref{eq:susceptibles}-\eqref{eq:infected}-\eqref{eq:recovered}.
    \begin{enumerate}
        \item For $\omega\in\Omega$, $t_0 \geq 0$, and for any initial value $(S_{t_0}, I_{t_0}, R_{t_0}) \in\RR_+^{3\times K}$,
        there is an unique global solution $(S_t(\omega), I_t(\omega), R_t(\omega))_{t\geq t_0}$ in $\RR_+^{3\times K}$ of system \eqref{eq:susceptibles}-\eqref{eq:infected}-\eqref{eq:recovered}.
        \item $S$, $I$, and $R$ are $\FF^\cW$-adapted. 
    \end{enumerate}
\end{theorem}
\begin{proof}
 The differential system \eqref{eq:susceptibles}-\eqref{eq:infected}-\eqref{eq:recovered} is an ordinary differential equation (ODE) with stochastic coefficients sometimes abusively called stochastic differential equation (SDE).
    \begin{enumerate}
        \item Consider the "associated" deterministic ODE i.e. the one associated with a fixed sample path. For a fixed $\omega\in\Omega$, the system becomes

\begin{empheq}[left=\empheqlbrace]{align}
    \frac{\dr S_{k,t}(\omega)}{\dr t} &= -k Y_t(\omega) S_{k,t}(\omega)+ Y_t(\omega)\sum_{j=k+1}^{K} j S_{j,t}(\omega)\cdot b_{j j-k, t}(\omega) ,\label{eq:susceptibles w}\\
    \frac{\dr I_{k,t}(\omega)}{\dr t} &= -\gamma_{k,t}(\omega) I_{k,t}(\omega) + Y_t(\omega) \sum_{j=k}^{K} j S_{j,t}(\omega)\cdot b_{j k,t}(\omega) ,\label{eq:infected w}\\
    \frac{\dr R_{k,t}(\omega)}{\dr t} &= \gamma_{k,t}(\omega) I_{k,t}(\omega) ,\label{eq:recovered w}
\end{empheq}
where
\begin{align}
    Y_t(\omega) &= \frac{1}{N_0} \sum_{k=1}^{K} \beta_{k,t}(\omega) \cdot k I_{k,t}(\omega)\label{eq:force w}.
\end{align}
In the above equation, both the variables $S,I,R,Y$ and the parameters $\beta_1,\hdots,\beta_K, \gamma_1,\hdots,\gamma_K,\allowbreak b_{1 1},\hdots,b_{1 K},b_{2 1},\hdots,b_{K K}$ depend on $\omega$. However, we will not specify it further in order to lighten the notation. 
The system can be rewritten in the following form 
\begin{align*}
    (\dot{S}, \dot{I}, \dot{R}) = F (S, I, R, \omega),
\end{align*}
where $F(S, I, R, \omega)$ corresponds to the right hand side of the system \eqref{eq:susceptibles w}-\eqref{eq:infected w}-\eqref{eq:recovered w}. The function $F$ is locally Lipschitz continuous, for any given initial value $(S_{t_0}, I_{t_0}, R_{t_0}) \in\RR^{3K}_+$:  according to Picard-Lindelöf theorem, there is an unique local solution $(S_t, I_t, R_t) $ on $t\geq t_0$, for any $t_0 \geq 0$.
From \eqref{eq:infected w}, we have
\begin{align*}
    \frac{\dr I_{k,t}}{\dr t} &= -\gamma_{k,t} I_{k,t} + Y_t \sum_{j=k}^{K} j S_{j,t}\cdot b_{jk,t}
    \geq -\gamma_{k,t} I_{k,t},
\end{align*}
then $\frac{\dr I_{k,t}}{I_{k,t}}
    \geq -\gamma_{k,t} \dr t$, which implies that 
\begin{align*}
    I_{k,t} \geq I_{k,t_0} \exp{\left(-\int_{t_0}^{t}\gamma_{k,s} \dr s\right)} \geq 0.
\end{align*}
From \eqref{eq:recovered w}, by noticing that $\frac{\dr R_{k,t}}{\dr t} \geq - \alpha_{k,t} R_{k,t}$, and from \eqref{eq:susceptibles w}, by noticing that $Y_t \leq 1$ and $\frac{\dr S_{k,t}}{\dr t} 
    \geq -k S_{k,t}$,
we also conclude that $S_{k,t} \geq 0$ and $R_{k,t} \geq 0$. However, from Remark~\ref{rem:total population}, the average size of firm $N$ is constant and equals to $N_0$, 
consequently the local solution can be extended to a global one and belongs to $\RR_+^{3\times K}$.
    
    \item When a solution exists, the second item is straightforward. Since  each  coefficient $\varphi\in\Psi$ satisfies Lipschitz and linear growth conditions, there exists an unique strong solution, namely,
        \begin{equation}\label{eq:varphi solution}
        \begin{split}
            \tilde\varphi_t &= \varphi_0 e^{-\kappa_\varphi t} + \mu_\varphi\left( 1 - e^{-\kappa_\varphi t} \right) + \Sigma_\varphi \int_0^t e^{-\kappa_\varphi (t-s)} \sqrt{ \tilde\varphi_s} \, d\cW_{\varphi,s},\quad\text{for all } t\geq 0
        \end{split}
        \end{equation}
        which is $\FF^\cW$-adapted. 
        Therefore all coefficients $\varphi$  of the system \eqref{eq:susceptibles}-\eqref{eq:infected}-\eqref{eq:recovered} are $\FF^\cW$-adapted, whose unique solution $(S,I R)$ is also $\FF^\cW$-adapted.
    \end{enumerate}
\end{proof}
To analyze the spread of the cyber-attack, we can introduce additional quantities such as the prevalence which is a non-negative process representing the fraction of
recovered subunits,
\begin{align}\label{eq:prevalence}
    \mathfrak{P} := \frac{1}{N_0} \sum_{k=1}^{K} k R_k
\end{align}
and the peak, a random variable representing the maximum of total number of infected subunits,
\begin{align}\label{eq:peak}
    \Jf :=  \max_{t\geq 0}\sum_{k=1}^{K} k I_{k,t}.
\end{align}
Also
in epidemiology modeling literature, the basic reproduction number is the number of new
infections produced by a typical infective individual in a population at a  Disease Free Equilibrium (DFE), the 
 state at which a population remains in the absence of disease. It is used to determine if the disease always dies out (after some time) or if it persists (around an endemic equilibrium). When the model parameters are constant, \cite{doenges2024sir} shows that the basic reproduction number $R_0$ is 
 \begin{align}\label{eq:R_0}
      R_0 = \sum_{i=1}^{K} i \frac{\beta_i }{\gamma_i}\sum_{m=i}^{K}\frac{m S_m}{N_0}\cdot b_{m,i}.
 \end{align}
 Moreover, some works such as \cite{van2002reproduction,guo2006global} show that, if the model admits a DFE, then the latter is locally asymptotically stable if $R_0 < 1$. Similar results exist for the stochastic case but where the randomness is introduced by using a system of SDE instead of a system of ODE (see \cite{ji2011multigroup}). 
 The following theorem provides a  basic reproduction number and a local stability condition for  the dynamic contagion model developed in this paper. 
\begin{theorem}[Necessary condition for the local stability]\label{prop:basic reproduction number}
Let $t\geq 0$. If 
\begin{align}\label{eq:basic reproduction number}
    \cR_{max,t}:=\left[\max_{1\leq k\leq K}\frac{\beta_{k,t} }{\gamma_{k,t}}\right] \sum_{i=1}^{K} i\sum_{m=i}^{K}\frac{m S_{m,t} }{N_0}\cdot b_{m i,t} < 1,
\end{align}
then $\frac{\dr\log{\Jf_t}}{\dr t} < 0.$
\end{theorem}
\begin{proof}
Let $t\geq 0$, from \eqref{eq:peak}, we have 
\begin{align*}
    \Jf_t' &= \sum_{k=1}^{K} k I_{k,t}'=\sum_{k=1}^{K} k \left(-\gamma_{k,t} I_{k,t} + Y_t \sum_{j=k}^{K} j S_{j,t}\cdot b_{jk,t}\right).
\end{align*}
    By applying the Itô formula

\begin{align*}
        \frac{\dr\log{\Jf_t}}{\dr t} &= \frac{\dr\Jf_t}{\dr t} \frac{1}{\Jf_t}  = -\frac{1}{\Jf_t}\sum_{k=1}^{K} k \gamma_{k,t} I_{k,t} + Y_t \sum_{k=1}^{K} k\sum_{j=k}^{K} j S_{j,t}\cdot b_{jk,t}\\
    &= -\frac{1}{\Jf_t} \left(\sum_{k=1}^{K} k \gamma_{k,t} I_{k,t}-\frac{1}{N_0} \sum_{k=1}^{K} \beta_{k,t} \cdot k I_{k,t} \sum_{i=1}^{K} i\sum_{m=i}^{K} m S_{m,t}\cdot b_{m i,t}\right)\\
    &=-\frac{1}{\Jf_t}\sum_{k=1}^{K} k \gamma_{k,t} I_{k,t}\left(1 - \frac{\beta_{k,t} }{\gamma_{k,t}} \sum_{i=1}^{K} i\sum_{m=i}^{K}\frac{m S_{m,t}}{N_0}\cdot b_{m i,t}\right)\\
    &=-\sum_{k=1}^{K} \frac{k \gamma_{k,t} I_{k,t}}{\Jf_t}\left(1 - \cR_{k,t} \right),
    \end{align*}
    where we denote 
    \begin{align}\label{eq:Rk}
        \cR_{k,t} := \frac{\beta_{k,t} }{\gamma_{k,t}} \sum_{i=1}^{K} i\sum_{m=i}^{K}\frac{m S_{m,t}}{N_0}\cdot b_{m i,t}.
    \end{align}
With $\cR_{max}$ defined in \eqref{eq:basic reproduction number}, we have
    \begin{align*}
        \cR_{k,t} \leq \cR_{max,t}.
    \end{align*}
   Therefore $1- \cR_{max,t} \leq 1-\cR_{k,t}$  and  since $\frac{-k \gamma_{k,t} I_{k,t}}{\Jf_t} \leq 0$, we have
    \begin{align*}
        \frac{\dr\log{\Jf_t}}{\dr t} &\leq -\sum_{k=1}^{K} \frac{k \gamma_{k,t} I_{k,t}}{\Jf_t}\left(1- \cR_{max,t}\right),
    \end{align*}
which is negative when $\cR_{max,t} < 1$. 
\end{proof}
In addition, if $\cR_\infty := \max_{t\geq 0}\{\cR_{max,t}\}<1$, we can easily show that there exists $h<0$ such that
\begin{align}\label{eq:expon decay}
    \limsup_{T\to+\infty} \frac{\log{\Jf_t}}{T} \leq h.
\end{align}
\begin{remark} We can also make the following remarks:
    \begin{itemize}
        \item In our framework, $\cR_{max,t}$ (respectively $\cR_{\infty}$) represents the local (respectively global) basic reproduction number. 
        \item In addition to the model parameters, both $\cR_{max,t}$ and $\cR_{\infty}$ can depend on the trajectories of susceptible subunits $(S_{m,t})_{1\leq m\leq K, t\geq 0}$. When $\frac{\beta_{k,t} }{\gamma_{k,t}}$ is independent of $k$ and $t$ and $b_{m i,t}$ independent of $t$ as in \cite{doenges2024sir}, we retrieve exactly the expected $R_0$ given in \eqref{eq:R_0}.
        \item It is obvious that $\Jf \geq 0$ and \Cref{theorem:existence and unicity} implies that $\Jf' < 0$ therefore the total number of infected subunits strictly decreases locally  with time, or in other words, the epidemic is set to disappear.
        \item Better, \eqref{eq:expon decay} means that $\Jf$ almost surely converges exponentially to $0$ or the epidemic disappears exponentially.
        \item These results mean that the contagion-free equilibrium point, if it exists, is locally asymptotically stable. 
    \end{itemize}
\end{remark}

\section{From the stochastic SIR to the impact on an insurance portfolio} \label{sec:linked}
Once the cyber contagion model is described, we are interested in  the following questions concerning the distribution of the  subunits'  and firm's sizes, as well as the impact of the cyber contagion on the firm size growth, on a cyber insurance portfolio and more generally on 
the economy growth.
Recall from \eqref{eq:firm revenue} and \eqref{eq:sub-unit dynamics}, the revenue $Z_{i,t}$ of firm $i\in\OneN$ at time $t\geq 0$, satisfies 
\begin{align*}
    Z_{i,t} = \sum_{j=1}^{K_i} z_{ij,t}
     = \sum_{j=1}^{K_i}  \left(1-\pi_{ij} \bOne_{ \tau_{ij} \leq t < \tau_{ij} + \delta_{ij}} \right)\overline z_{ij,t},
\end{align*}
where $\tau_{ij}$ is the first jump time of the point process $N_{ij}$ and $\delta_{ij}$ is the random recovery time.
The arrival of  a cyber-attack can come either from outside or from a sister subunit. We thus have the following assumption.
\begin{assumption}\label{ass:lambda V_ij}
    For $(i,j)\in\Jj$ and $t\geq 0$, we assume 
    \begin{align*}
        N_{ij,t} := N_{ij,t}^{0} + \bOne_{\left\{N_{ij,\tau_i^-}^{0}=0\right\}} \bOne_{\left\{t\ge \tau_i\right\}} U_{ij},
    \end{align*}
where 
\begin{itemize}
    \item $(N_{ij}^{0})_{ij}$ is a sequence of independent Cox processes with common intensity $Y$ defined in \eqref{eq:force}.
    \item $\tau_i$ is the stopping time defined as \begin{align}\label{eq:tau_i}
        \tau_i &:= \inf\left\{t>0:\, \sum_{k=1}^{K_i} N_{ik,t}^{0}\ge 1\right\}.
    \end{align}
    \item $(U_{ij})_{ij}$ is a sequence of iid Bernoulli random variables with probability $a_{\tau_i}$.
\end{itemize}
\end{assumption}
In other words, for $(i,j)\in\Jj$ and $t\geq 0$, \Cref{ass:lambda V_ij} states that:
    \begin{itemize}
        \item When the infection of subunit $j$ of firm $i$ comes from outside between $t$ and $t+\dr t$, its probability is 
        $Y_t\dr t$
        where $Y_t$ is the force of infection defined in \eqref{eq:force}. This implies that  $\EE[N_{ij,t+\dr t}^{0} - N_{ij,t}^{0}|\cF^\cW_t] = Y_t\dr t$ recalling that $Y$ is $\FF^\cW$-adapted (with $\FF^\cW$ given in \eqref{eq:filtration F}).
        \item When there exists at least one sister subunit  of $j$ infected from outside at time $t$ i.e. $K_i > 1$ and  $\sum_{k=1}^{K_i} N_{ik,t}^{0}\ge 1$, and if the subunit $j$ is not infected from outside at the time $\tau_i$ i.e. $N_{ij,\tau_i^-}^{0}=0$, then it may have an internal contagion $U_{ij}$ with probability $a_{\tau_i}$ where $a$ is the process defined in Assumption~\ref{ass:SIR}. 
    \end{itemize}
    For each firm $1\leq i\leq H$, if $K_i > 1$, the vector of point processes $(N_{ij})_{1\leq j\leq K_i}$ is strongly correlated. We describe in the following proposition the marginal law of first jump of subunit $j$, namely $\tau_{ij}$ defined in \eqref{eq:first arrival time}.

\begin{proposition}\label{prop:survival tau_ij}
    For $(i,j)\in\Jj$, if $K_i\geq 2$, then the conditional marginal CDF function of $\tau_{ij}$ is,  \\for  $0\leq u\leq t\leq T$
    \begin{equation}\label{eq:marginal CDF tau_ij}
    \begin{split}
        F_{ij}^{|t} (u):= \PP[\tau_{ij} \leq u|\cF_t^\cW]=1-e^{-K_i \Lambda_u} - (K_i-1) e^{-\Lambda_u}\int_{0}^{u}  Y_s e^{-(K_i-1) \Lambda_s} (1-a_s)\dr s,
    \end{split}
    \end{equation}
 where $\Lambda_t := \int_{0}^{t} Y_s\dr s$.
\end{proposition}    
\begin{proof}
    Let $(i,j)\in\Jj$ and $0\leq u\leq t\leq T$. If $K_i = 1$, there is no internal contagion. We directly have 
    \begin{equation*}
    \begin{split}
        \PP[\tau_{ij} > u|\cF^\cW_t] = \EE\left[\bOne_{N_{ij,u}^{0} = 0}|\cF^\cW_t\right]= \EE\left[\EE\left[\bOne_{N_{ij,u}^{0} = 0}|\cF^\cW_u\right]|\cF^\cW_t\right] = \EE\left[e^{-\Lambda_u}|\cF^\cW_t\right] = e^{-\Lambda_u},
    \end{split}
    \end{equation*}
   since $Y$ is $\FF^\cW$-adapted.\\
    If $K_i \geq 2$,  we decompose  the set $\{\tau_{ij} > u\}$ in 2 disjoint sets $  (\{N_{ij,u}^{0} = 0 \}\cap \{\tau_{i} > u\})$ and $(\{N_{ij,u}^{0} = 0 \}\cap \{\tau_{i} \leq u\}\cap \{U_{ij} = 0 \})$. On the one hand
    \begin{align*}
        \PP[\{N_{ij,u}^{0} = 0 \}\cap \{\tau_{i} > u\} | \cF^\cW_t]
        & = \PP\left[\cap_{k=1}^{K_i} \{N_{ik,u}^{0} = 0 \}| \cF^\cW_t\right]= e^{-K_i \Lambda_u},
    \end{align*}
    where the last equality comes from the fact $\left(N_{ik}^{0}\right)_{ik}$ are iid conditionally to $\FF^\cW$. On the other hand,
    \begin{align*}
        \PP\left[\{N_{ij,u}^{0} = 0 \}\cap \{\tau_{i} \leq u\}\cup \{U_{ij} = 0 \}| \cF^\cW_t\right]&=\PP\left[N_{ij,}u^{0} = 0,\tau_{i} \leq u,U_{ij} = 0| \cF^\cW_t\right]\\
        & = \int_{0}^{u} \PP\left[N_{ij,u}^{0} = 0,\tau_{i} \in\dr s| \cF^\cW_t\right] (1-a_s).
    \end{align*}
    $\PP\left[N_{ij,u}^{0} = 0,\tau_{i} \in\dr s| \cF^\cW_t\right] = \PP\left[N_{ij,u}^{0} = 0| \cF^\cW_t\right] \PP\left[\tau_{i} \in\dr s| N_{ij,u}^{0} = 0, \cF^\cW_t\right]$, with 
  $ \PP\left[N_{ij,u}^{0} = 0| \cF^\cW_t\right] = e^{-\Lambda_u}$, and
    \begin{align*}
        \PP\left[\tau_{i} \in\dr s| N_{ij,u}^{0} = 0, \cF^\cW_t\right] &= \frac{\dr}{\dr s} \PP\left[\tau_{i} \leq s| N_{ij,u}^{0} = 0, \cF^\cW_t\right]\\
        &= -\frac{\dr}{\dr s} \PP\left[\cap_{k=1,k\neq j}^{K_i} \{N_{ik,s}^{0} = 0 \}| \cF^\cW_t\right]\\
        &= -\frac{\dr}{\dr s} e^{-(K_i-1) \Lambda_s}\\
        &= (K_i-1) Y_s e^{-(K_i-1) \Lambda_s}.
    \end{align*}
    Therefore, 
    \begin{align*}
        \PP\left[\{N_{ij,u}^{0} = 0 \}\cap \{\tau_{i} \leq u\}\cup \{U_{ij} = 0 \}| \cF^\cW_t\right] = \int_{0}^{u} (K_i-1) Y_s e^{-(K_i-1) \Lambda_s} (1-a_s) e^{-\Lambda_u} \dr s.
    \end{align*}
    Finally
    \begin{align*}
        \PP[\tau_{ij} > u | \cF^\cW_t] &= e^{-K_i \Lambda_u}+ (K_i-1) e^{-\Lambda_u}\int_{0}^{u}  Y_s e^{-(K_i-1) \Lambda_s} (1-a_s)  \dr s,
    \end{align*}
    and the conclusion follows.
\end{proof}

\begin{remark}\label{rem:lambda V_ij}
    Let $t\geq 0$. From \eqref{eq:marginal CDF tau_ij}, we can make the following remarks:
    \begin{enumerate}
    \item The right term does not depend on the subunit so that the marginal distributions of first jump of all  subunits of the same firm are identical. We will denote $F_i^{|t}$ the $\cF_t^\cW$-cumulative distribution function, instead of $F_{ij}^{|t}$.
    \item The form of $F_i^{|t}$ implies that the $\cF_t^\cW$-conditional probability that a subunit becomes infected increases with the probability of internal contagion.
     \item For a fixed firm $i$ at a fixed date $t$ and for a fixed trajectory of $\cW$,  if the function $K_i \mapsto (K_i-1)\Lambda_s - 1$ is negative for all $0\leq s\leq u$, then the function $K_i \mapsto F_i^{|t}(u)$ is non-decreasing. 
     In fact, by differentiating the function $K_i \mapsto F_{i}^{|t}(u)$, we get
    \begin{align*}
        \frac{\dr F_{i}^{|t}(u)}{\dr K_i} &= K_i e^{-K_i \Lambda_u} -  e^{-\Lambda_t}\int_{0}^{u}  Y_s e^{-(K_i-1) \Lambda_s} (1-a_s)  \dr s+ (K_i-1) e^{-\Lambda_u}\int_{0}^{u} \Lambda_s Y_s e^{-(K_i-1) \Lambda_s} (1-a_s)  \dr s,\\
        &= K_i e^{-K_i \Lambda_u} +  e^{-\Lambda_u}\int_{0}^{u} [(K_i-1)\Lambda_s-1] Y_s e^{-(K_i-1) \Lambda_s} (1-a_s)  \dr s.
     \end{align*}   
     This means that the $\cF_t^\cW$-probability that a subunit becomes infected increases with the number of subunits. This is obvious in particular when the process $(\Lambda_u)$ is bounded by $\frac{1}{K-1}$.
        \item For all firms with the same size (number of subunits), all their subunits have the same $\cF_t^\cW$-conditional probability to be infected at a given time.
        \item Let $1\leq i\leq H$ and $0\leq u\leq t$. Because
        \begin{align*}
            \PP\left[\cap_{k=1}^{K_i} \{N_{ik,u} = 0 \}| \cF^\cW_t\right] = \PP\left[\cap_{k=1}^{K_i} \{N_{ik,u}^{0} = 0 \}| \cF^\cW_t\right] =e^{-K_i\Lambda_u},
        \end{align*}
        the $\cF_t^\cW$-conditional probability that at least one subunit of firm $i$ gets infected before time $u$ is $1-e^{-K_i\Lambda_u}$ which increases with $K_i$. Therefore, the vulnerability of firms increases with their size.
    \end{enumerate}
\end{remark}
The marginal $\cF_t^\cW$-conditional CDF $(F_{i})$ are obviously $\cF_t^\cW-$measurable because $a$ and $Y$ are. In order to fully determine the firm revenue (see \eqref{eq:firm revenue}) and firm claims (see \eqref{eq:explicit c_i}) under cyber-attack, we need to characterize the random recovery time $(\delta_{ij})$ introduced in Assumption~\ref{ass:sub-unit dynamics}. We deduce $(\delta_{ij})$  from  the recovery rate process $\gamma_{i}$  in the SIR model \eqref{eq:recovered} as follows.
\begin{assumption}\label{ass:delta def}
    For $(i,j)\in\Jj$, we assume that 
    \begin{align}\label{eq:delta def}
        \delta_{ij} = \frac{1}{\gamma_{i,\tau_{ij}}}.
    \end{align}
\end{assumption}
\noindent This is motivated by the fact that when the recovery times of  infected individuals are modeled as independent exponentially distributed
random variables, the recovery rate is equal to the inverse of the expected recovery time (see \cite{doenges2024sir}).\\

Moreover, we define the filtration~$\mathbb{G}=(\cG_t)_{t\geq 0}$ by  

\begin{align}\label{eq:filtration G}
    \cG_{t} = \sigma\left(\cF^B_t \cup \cF^\cW_t \cup \sigma\left\{\pi_{ij,s}, N_{ij,s}: s\in[0, t]\text{ and } (i,j)\in\Jj\right\}\right),
\end{align}
where $\cF^B_t$ and $\cF^\cW_t$ are defined in \eqref{eq:filtration Fb} and in \eqref{eq:filtration F} respectively. 

 \begin{remark}It is straightforward that
the subunit revenue $(z_{ij})_{(i,j)\in\Jj}$, the total revenue of the portfolio $\mathbf{O}$, and the firm's claims $(c_{i,t})_{1\leq i\leq H}$, defined respectively in \eqref{eq:sub-unit dynamics}, \eqref{eq:whole economy}, and \eqref{eq:explicit c_i} are $\GG$-adapted.
 \end{remark}

With the knowledge of $(\tau_{ij})_{ij}$ and $(\delta_{ij})_{ij}$, it is now possible to compute the costs of cyber-events at the firm level $(\cc_i)_{i}$ from \eqref{eq:explicit c_i} and at the portfolio level $\aep$ from \eqref{eq:portfolio measure}. But let us note that other insurance portfolio measures can be studied, and in particular, the impact of reaction and remediation measures as in \cite{hillairet2021propagation}.

From \eqref{eq:explicit c_i}, we can identify  4 sources of randomness in the cost functions $\cc$, $\CC$, or $\aep$. The first one is the Brownian motion  $B$ governing  the economic risk and represented by $\FF^B$. The second one is the random coefficients of SIR represented by $\FF^\cW$, the third one is the random arrivals $(\tau_{ij})$ of cyber-attacks, and the last one is their random severity $\pi$. 
In fact, all these quantities have a "systemic part" which is $\FF^\cW$-adapted  coming from the cyber-attack and an "idiosyncratic/economic part" which are $(N_{ij})$, $(\pi_{ij})$, $(B_{ij})$. 
   We can assume that the insurer mutualizes well  the idiosyncratic/economic risks, as $(B_{ij})_{(i,j)\in\Jj}$ are independent between firms. Therefore we may focus on  the systemic cyber risk component and consider $\FF^\cW$ as the only source of randomness by taking the expectation conditional to the systemic risk. This leads to the notion of  conditional $\aep_T^\star$, without the idiosyncratic risk: 
\begin{align}\label{eq:proxy EAP}
    \aep_T^\star(x) &:= 1 - \EE_P\left[F_{\CC_T^\star}^{(P)}(x) \right],
\end{align}
   where 
   $\CC_T^\star$ is the total expected loss conditionally to the  random SIR  factors up to $t$, $\cF^\cW_t$
\begin{align*}
    \CC_T^\star := \sum_{i=1}^{H}\left( \int_{0}^{T}\EE\left[c_{i,t}|\cF^\cW_t\right]\dr t\right),
\end{align*}
and  $c_{i,t}$ is defined in \eqref{eq:explicit c_i}. The following proposition gives an explicit expression of $\EE\left[c_{i,t}|\cF^\cW_t\right]$.

\begin{proposition}
    For $1\leq i\leq H$ and $t,u\geq 0$,
    
    \begin{align}
        \EE\left[c_{i,t}|\cF^\cW_t\right] = \left(\pi_\star\sum_{j=1}^{K_i} z_{ij,0} e^{\mu_{ij} t}\right) \int_{0}^{t} \bOne_{\left\{\gamma_{i,u}< \frac{1}{t-u} \right\}} \dr F_i^{|t}(u),
    \end{align}
    
    where $F_{i}$ is defined in \eqref{eq:marginal CDF tau_ij}.
\end{proposition}
\begin{proof}
    Let $1\leq i\leq H$ and $t,u\geq 0$. From the explicit expression of the claims at $t$ given in \eqref{eq:explicit c_i}, we have:
    \begin{align*}
        \EE\left[c_{i,t}|\cF^\cW_t\right]= \EE\left[\sum_{j=1}^{K_i} \bar z_{ij,t}\pi_{ij} \bOne_{ \tau_{ij} \leq t < \tau_{ij} + \delta_{ij}}|\cF^\cW_t\right].
    \end{align*}
    Given that $(\pi_{ij})$ are independent of $\cF^\cW_t$ with $\EE[\pi_{ij}] = \pi_\star$ and that, from Assumptions \ref{ass:sub-unit dynamics} and \ref{ass:SIR}, $\EE[z_{ij,t}|\cF^\cW_t] = \EE[z_{ij,t}] = z_{ij,0} e^{\mu_{ij} t}$, we have
     \begin{align*}
        \EE\left[c_{i,t}|\cF^\cW_t\right]= \pi_\star\sum_{j=1}^{K_i} z_{ij,0} e^{\mu_{ij} t} \EE\left[\bOne_{ \tau_{ij} \leq t < \tau_{ij} + \frac{1}{\gamma_{i,\tau_{ij}}}}\middle|\cF^\cW_t\right],
    \end{align*}
    recalling from \eqref{eq:delta def} that $\delta_{ij} = \frac{1}{\gamma_{i,\tau_{ij}}}$.
    For each $1\leq j\leq K_i$, 
    \begin{align*}
        \EE\left[\bOne_{\left\{\tau_{ij} \leq t < \tau_{ij} + \frac{1}{\gamma_{i,\tau_{ij}}}\right\}}\middle|\cF^\cW_t\right] &=\EE\left[\bOne_{\left\{t-\frac{1}{\gamma_{i,\tau_{ij}}} < \tau_{ij} \leq t\right\}}\middle|\cF^\cW_t\right] \\
        &= \int_{0}^{+\infty} \bOne_{\left\{t-\frac{1}{\gamma_{i,u}} < u \leq t \right\}} \dr \PP(\tau_{ij} \leq u|\cF^\cW_t)\\
        &= \int_{0}^{t} \bOne_{\left\{\gamma_{i,u}< \frac{1}{t-u} \right\}} \dr \PP(\tau_{ij} \leq u|\cF^\cW_t).
    \end{align*}
    We conclude using  $\PP(\tau_{ij} \leq u|\cF^\cW_t) = F_{i}^{|t}(u)$ for $0\leq u\leq t$ as given in \Cref{prop:survival tau_ij}.
\end{proof}
\noindent After further straightforward calculations, we obtain

\begin{align}\label{eq:approx C_T}
    \CC_T^\star= \sum_{(i,j)\in\Jj} \frac{\pi_\star z_{ij,0}}{\mu_{ij}} \int_{0}^{T}   \left( e^{\mu_{ij} \min \left(T, u + \frac{1}{\gamma_{i,u}} \right)} - e^{\mu_{ij} u} \right) dF_i^{|u}(u).
\end{align}
Computing $\CC_T^\star$ is faster than $\CC_T$. Indeed, instead of simulating $|\mathcal J|$ random arrival times and integrating $|\mathcal J|$ geometric Brownian motion paths, we only need to evaluate $|\mathcal J|$ deterministic integrals. This reduction becomes particularly advantageous when the portfolio is large (i.e., when $|\mathcal J|$ or $H$ is big). 
In our simulations, the size of the portfolio being reasonable,  the distribution of $C_T$ remains computationally tractable to compute the exact AEP.     In \Cref{fig:AEP_curve}, we compare these exact AEP curves to those obtained for $C_T^\star$.

\section{Numerical analysis and discussion}\label{sec:simulations} 

This section describes the data for parameters calibration and estimation, outlines the calibration methods and comments on the estimated parameters, and finally presents the simulation methodology and discusses the results.

\subsection{The dataset}
\subsubsection{Firm's data}
The dataset \cite{pappers2025} provides information of $H=2,929$ firms located in the Ile-de-France region, in particular the sector of activity (using the NAF classification as in \cite{inseeNAF2025}) and the annual revenue (in \euro million) from 2010 to 2022 (that is for  $T^f = 13$ years). \Cref{tab:firm_characteristics} summarizes firms' characteristics -- including the count, mean, and standard deviation of revenue-- by sector.
\begin{table*}[ht!]
\centering\small
\begin{tabular}{|p{6cm}|r|r|r|}
\hline
\textbf{Sector} & \textbf{Number of firms} & \textbf{Average revenue} & \textbf{Std of revenue} \\ \hline 
Manufacturing & 339 & 7.13 & 12.18 \\
Technology & 87 & 5.85 & 10.09 \\
Retail/Wholesale & 788 & 9.91 & 18.11 \\
Healthcare / Education / Government & 65 & 2.03 & 3.54 \\
Consulting / Legal & 383 & 5.01 & 9.59 \\
Construction & 559 & 5.53 & 9.94 \\
Banking/Finance & 64 & 2.52 & 6.12 \\
Miscellaneous** & 644 & 3.73 & --- \\ \hline
\textbf{Total} & \textbf{2929} & \textbf{2.87} & \textbf{13.04} \\ \hline
\end{tabular}
\caption{Firms characteristics. (**) Includes Agriculture, Mining, Transport, and Real Estate.}
\label{tab:firm_characteristics}
\end{table*}

\subsubsection{Cyber contagion data}
Each quarter, the \textit{GuidePoint Security’s Research and 
Intelligence Team} publishes a  Ransomware and Cyber Threat Insights (see~\cite{guidepoint2025} for the year 2024). This report summarizes cyber crime incidents, mainly ransomware attacks, that occurred during the year. Data collected are obtained from publicly available resources, including the sites and blogs of threat groups themselves. \Cref{fig:Most_Impactful_Ransomware_2024} and \Cref{fig:Most_Impacted_Industries_2024} summarize, respectively, the most impactful ransomware groups and the most impacted industries in 2024. 
For this numerical analysis, we calibrate the cyber-episode on the  \textbf{LockBit} ransomware  which, according to \cite{guidepoint2025}, entered 2024 as the long-standing dominant ransomware group with the highest tempo by victim, but the group faced substantial disruption in the wake of February’s international operation Cronos (see \cite{europol2024lockbit}). The calibration and validation are done on the time interval from May 1, 2024 to July 31, 2024, that is $T =100$ days.

The database gives the total infected per industry in 2024 for all the ransomware groups, the total infected per group day for each group. However, the data do not specify the number of infected firms by size. We address this by proposing a proxy in \Cref{sec:calib_contagion}, based on the estimated firm size distribution.
\begin{figure}[ht!]
\centering
\begin{minipage}{0.48\textwidth}
\centering\small
\begin{tabular}{|l|c|c|}
\hline
\textbf{Industry} & \textbf{Victims} & \textbf{Rate} \\
\hline
Manufacturing & 610 & 20.68\% \\
Technology & 390 & 13.22\% \\
Retail/Wholesale & 335 & 11.36\% \\
Healthcare & 325 & 11.02\% \\
Consulting & 275 & 9.32\% \\
Construction & 235 & 7.97\% \\
Education & 205 & 6.95\% \\
Legal & 200 & 6.78\% \\
Government & 190 & 6.44\% \\
Banking/Finance & 185 & 6.27\% \\
\hline
\end{tabular}
\captionof{table}{Most impacted industries in 2024}
\label{tab:Most_impacted_industries_2024}
\end{minipage}
\hfill
\begin{minipage}{0.48\textwidth}
\centering
\includegraphics[width=\columnwidth]{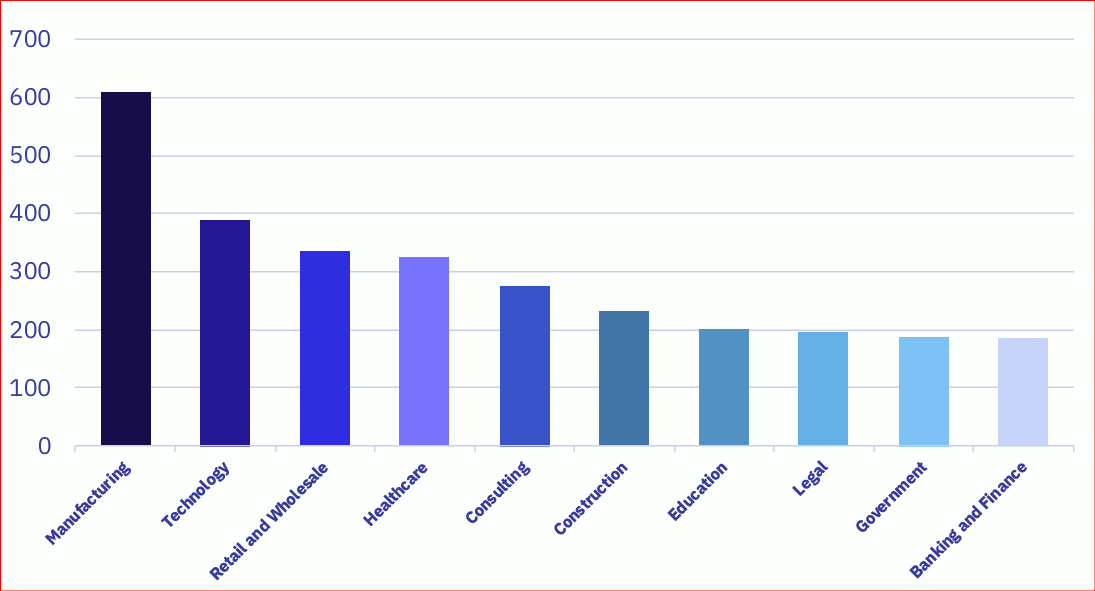}
\caption{Industries impacted (x-axis: sector, y-axis: victims)}
\label{fig:Most_Impacted_Industries_2024}
\end{minipage}
\end{figure}

\subsection{Calibration procedure}
\subsubsection{Calibration of the firms parameters}\label{sec:calib_firms_params}
Recall that there are $H = 2,929$ firms in the portfolio. For each firm of the portfolio, we have annual revenue data $Z_{i,t}$ between year $1$ (corresponding to 2010) and year $T^f =13$ (corresponding to 2023). However, since we do not have the number of subsidiaries and their revenue, we construct the following proxy for both variables.\\
The empirical average of revenue in the database is written as
\begin{align}
    \overline{Z}_t = \frac{1}{H} \sum_{i=1}^{H}Z_{i,t}.
\end{align}
We assume that company $i\in\OneN$ has a single subsidiary if its initial revenue $Z_{i,1}$ is below the initial sample average $\bar{Z}_1$; otherwise, it is assumed to have multiple subsidiaries. Formally,
\begin{align}
    K_i = \left\lceil \frac{Z_{i,1}}{\overline{Z}_1} \right\rceil \in\NN^*,
\end{align}
where $\lceil x\rceil$ is the smallest integer greater than or equal to $x$.
Then, assuming that all subsidiaries of a same firm have the same size, we have, for all $j\in\{1,\hdots,K_i\}$, 
\begin{align}
    z_{ij,t} = \frac{Z_{i,t}}{K_i}.
\end{align}
Then the estimation of the Black Scholes parameters  \eqref{eq:sub-unit dynamics wo cyber}, $(\sigma_{ij})_{(i,j)\in\Jj}$ and $(\mu_{ij})_{(i,j)\in\Jj}$, is given by
\begin{align*}
    \sigma_{ij}^{\text{year}} = \sqrt{\frac{1}{T^f-1}\sum_{t=1}^{T^f-1} (r_{ij,t} - r_{ij,*})^2}, \;
\text{ and } \;
    \mu_{ij}^{\text{year}} = r_{ij,*} +\frac{1}{2}(\sigma_{ij}^{\text{year}})^2,
\end{align*}
where for each subsidiary, 
$r_{ij,t} := \log{\frac{z_{ij,t+1}}{z_{ij,t}}}$ is the revenue growth between two consecutive years and $r_{ij,*} := \frac{1}{T^f-1} \displaystyle\sum_{t=1}^{T^f-1}r_{ij,t}$ is the average revenue growth over the period $T^f$. 
Finally, since cyber-events are tracked on a daily frequency, we convert the annual parameters to daily parameters by assuming 365 days:
\begin{align*}
    \sigma_{ij} = \frac{\sigma_{ij}^{\text{year}}}{\sqrt{365}} \quad
\text{and}\quad
    \mu_{ij}^{\text{year}} = \frac{\mu_{ij}^{\text{year}}}{365}.
\end{align*}
We use the firms' revenue dataset described in \Cref{tab:firm_characteristics} to determine the firms/subunit characteristics.\\

\Cref{tab:nb_firm_with_subunits} below indicates the number of firms per size (i.e. the number of subunits) and per sector. The  maximum firm size is $K=12$. 
\begin{table*}[ht!]
\centering\small
\begin{tabular}{|p{5.5cm}|r|r|r|r|r|r|r|r|r|r|r|r|} 
\hline
& \multicolumn{12}{c|}{\textbf{Firm size}} \\ \cline{2-13}
\textbf{Sector} & \textbf{1} & \textbf{2} & \textbf{3} & \textbf{4} & \textbf{5} & \textbf{6} & \textbf{7} & \textbf{8} & \textbf{9} & \textbf{10} & \textbf{11} & \textbf{12} \\ \hline 
Manufacturing & 249 & 51 & 20 & 7 & 2 & 2 & 1 & 1 & 3 & 1 & 2 & 0 \\
Technology & 67 & 5 & 5 & 4 & 2 & 0 & 3 & 0 & 1 & 0 & 0 & 0 \\
Retail/Wholesale & 619 & 80 & 42 & 18 & 11 & 7 & 4 & 1 & 1 & 1 & 2 & 2 \\
Healthcare/Education/Government & 47 & 12 & 1 & 1 & 0 & 1 & 0 & 2 & 1 & 0 & 0 & 0 \\
Consulting / Legal & 317 & 31 & 22 & 6 & 1 & 2 & 1 & 0 & 1 & 2 & 0 & 0 \\
Construction & 433 & 62 & 30 & 11 & 7 & 4 & 3 & 5 & 2 & 1 & 0 & 1 \\
Banking/Finance & 57 & 2 & 0 & 2 & 1 & 1 & 0 & 0 & 1 & 0 & 0 & 0 \\
Miscellaneous** & 474 & 69 & 18 & 12 & 14 & 3 & 1 & 4 & 1 & 2 & 0 & 1 \\ \hline
\textbf{Total} & \textbf{2263} & \textbf{312} & \textbf{138} & \textbf{61} & \textbf{38} & \textbf{20} & \textbf{13} & \textbf{13} & \textbf{11} & \textbf{7} & \textbf{4} & \textbf{4} \\ \hline
\end{tabular}
\caption{Number of firms per size and per sector}
\label{tab:nb_firm_with_subunits}
\end{table*}
\Cref{fig:histogram_sub_unit} and \Cref{fig:zipf_number_sub_unit} show that the firm size is distributed according to Zipf law (see \eqref{eq:pareto sub-unit}); ; the fitted parameters are $\af = 1.759$ and $q = 0.784$. A small number of very large firms account for a large share of total mass, therefore aggregate quantities -- such as AEP -- are therefore sensitive to extreme firms.
\begin{figure}[ht!]
  \centering
  \begin{subfigure}[b]{0.495\textwidth}
    \centering
    \includegraphics[width=\textwidth]{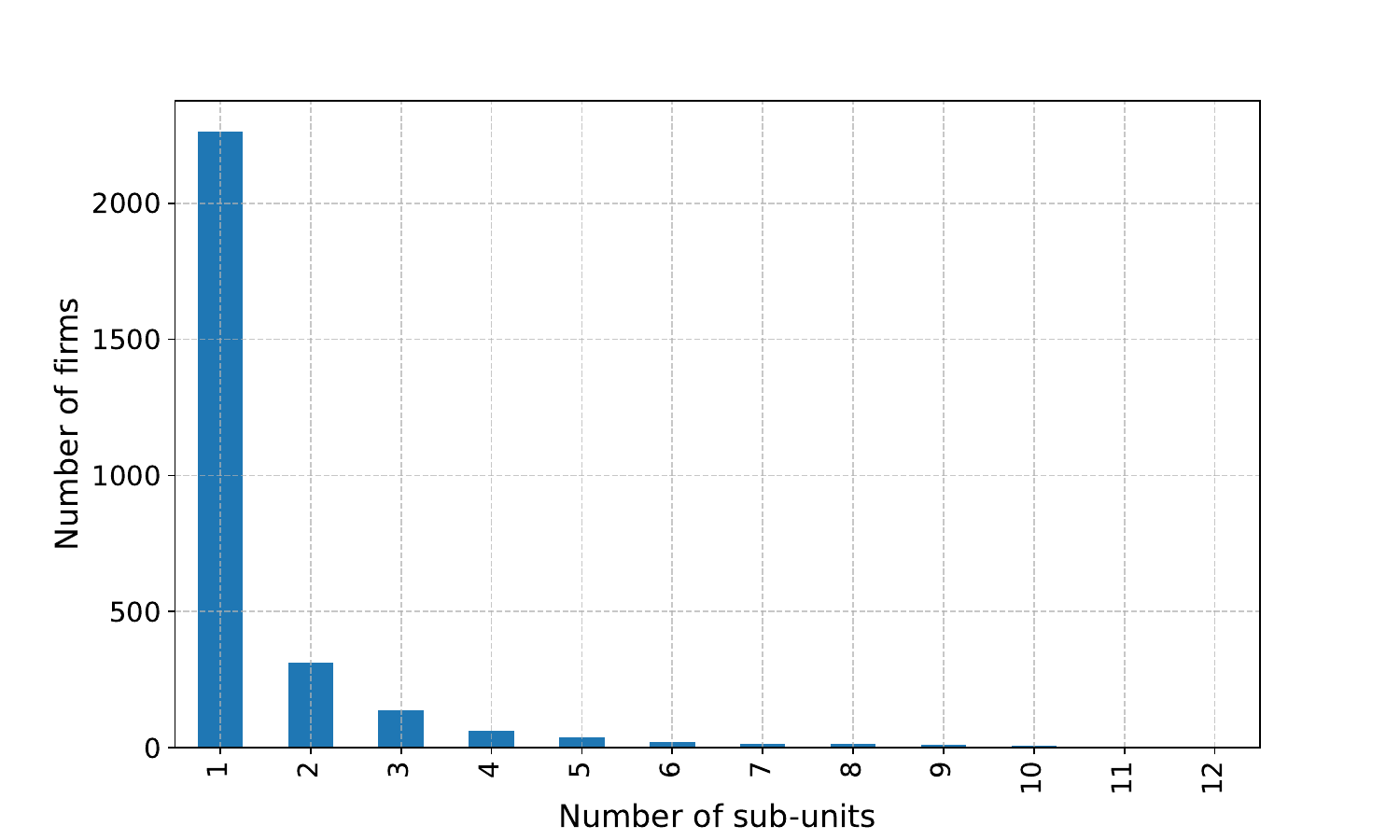}
    \caption{Histogram}
    \label{fig:histogram_sub_unit}
  \end{subfigure}
  \hfill
  \begin{subfigure}[b]{0.495\textwidth}
    \centering
    \includegraphics[width=\textwidth]{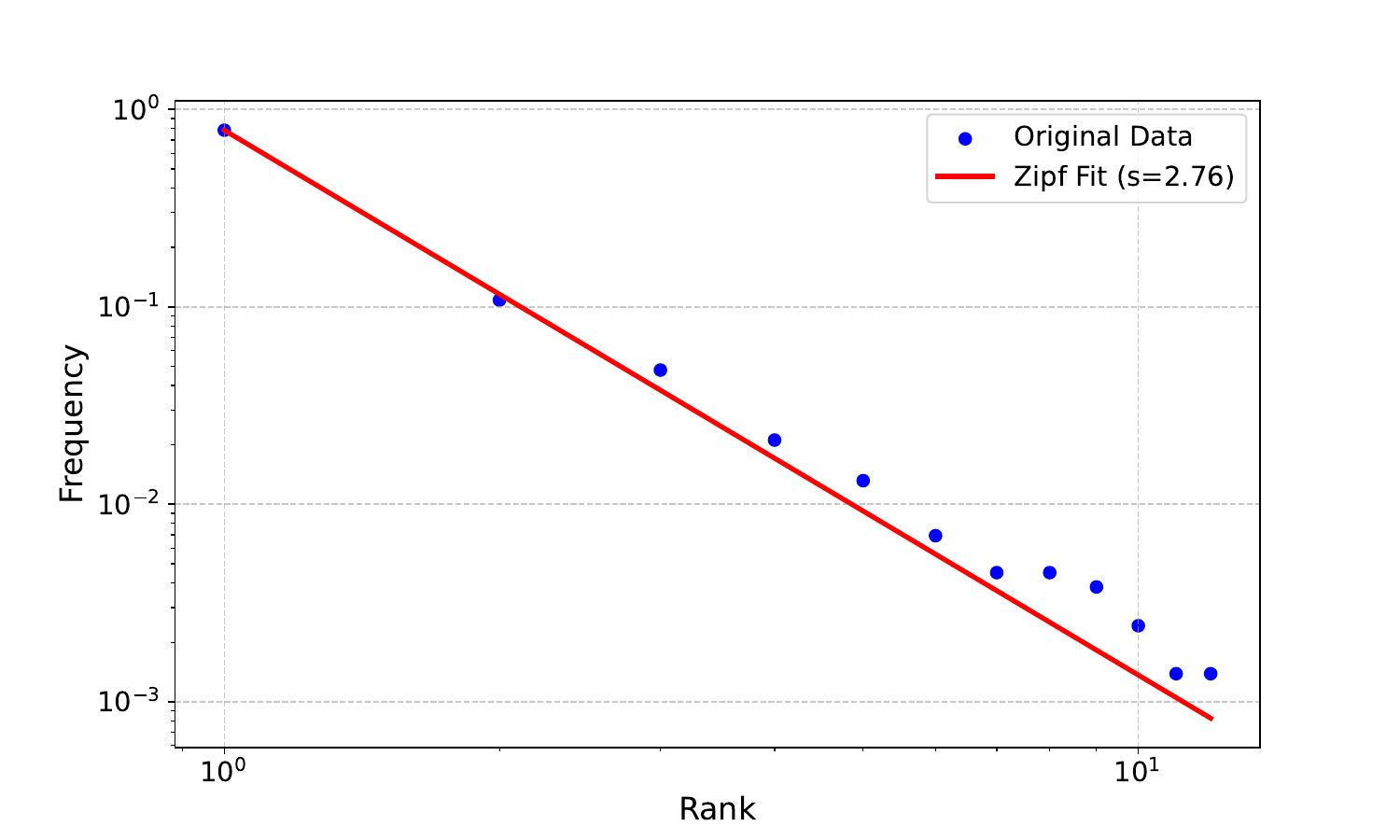}
    \caption{Fit of the Zipf law}
    \label{fig:zipf_number_sub_unit}
  \end{subfigure}
  \caption{Firm size}
\end{figure}

In \Cref{tab:firm_characteristics_size}, we summarize the average drift, the volatility, and the initial revenue, for each size.
\begin{table*}[ht!]
    \centering\footnotesize
    \begin{tabular}{|l|r|r|r|r|r|r|r|r|r|r|r|r|}
\hline
$K_i$ & 1 & 2 & 3 & 4 & 5 & 6 & 7 & 8 & 9 & 10 & 11 & 12 \\
\hline
$z_{ij}^0$ (in \euro million)   & 4.70 & 5.99 & 6.99 & 6.69 & 7.30 & 8.98 & 7.49 & 5.40 & 6.70 & 6.61 & 10.03 & 9.82 \\
\hline
$\sigma_{ij} \times 10^3$ & 12.69 & 12.30 & 10.91 & 9.70 & 16.16 & 8.16 & 19.03 & 26.31 & 8.64 & 12.71 & 4.92 & 4.79 \\
\hline
$\mu_{ij}\times 10^3$    & 0.19 & 0.37 & 0.28 & 0.26 & 0.56 & 0.24 & 1.12 & 1.45 & 0.36 & 0.28 & 0.33 & 0.27 \\
\hline
\end{tabular}
    \caption{subunits' characteristics}
    \label{tab:firm_characteristics_size}
\end{table*}

\subsubsection{Calibration of the SIR parameters}\label{sec:calib_contagion}

To build a proxy for the number of infected firms per size, we first divide the number of attacks per sector by the total number of attacks to get the attack rate per sector in 2024 (see Table~\ref{tab:Most_impacted_industries_2024}). By multiplying the latter by the total number of LockBit infections, we obtain the average number of attacks per sector. Then, based on the distribution of subsidiaries across sectors in \Cref{tab:nb_firm_with_subunits}, we derive an estimate of the number of infected firms of varying sizes (see \Cref{fig:Historical_infected}). 
\begin{figure}[ht!]
    \centering
    \includegraphics[width=0.79\linewidth]{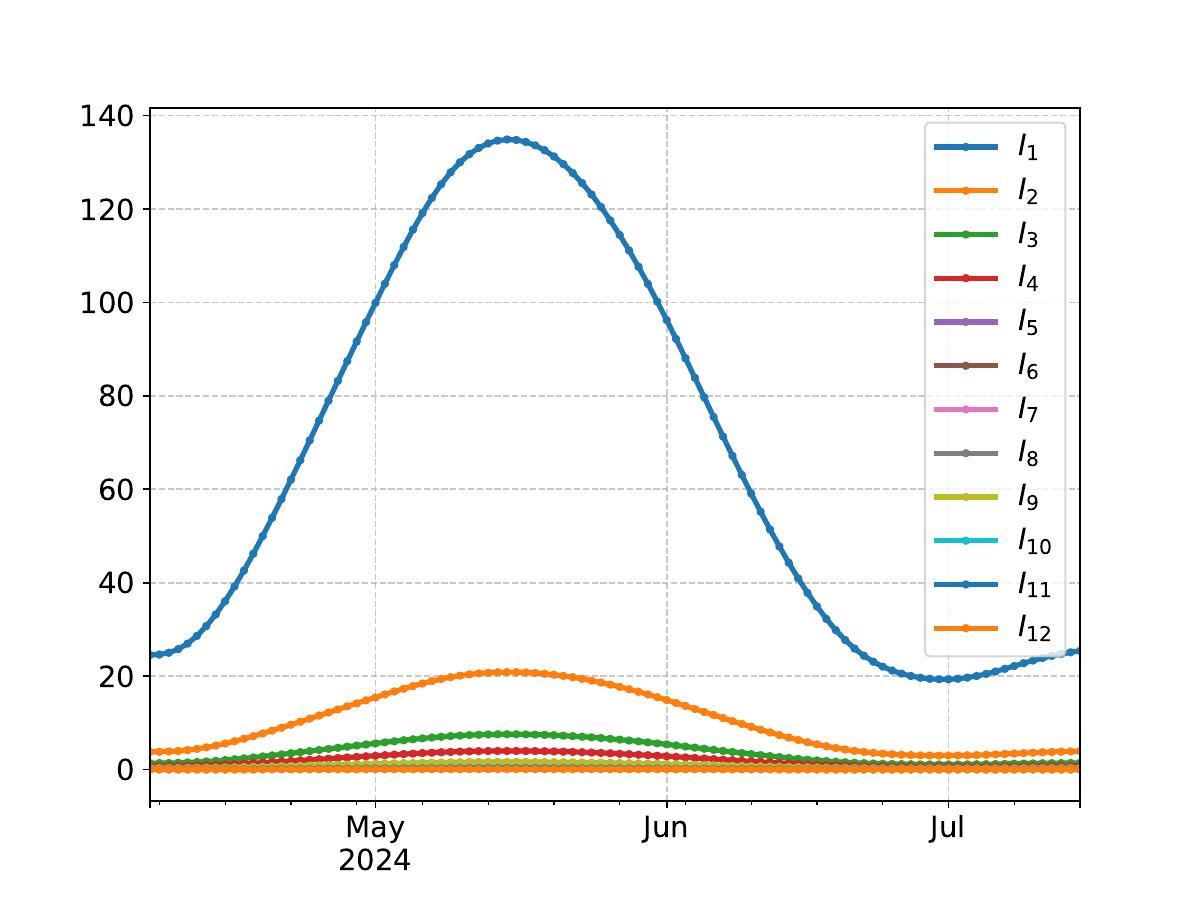}
    \caption{Number of daily new firms infected per size each day in May  - June - July 2024 ($I_k$)}
    \label{fig:Historical_infected}
\end{figure}

To calibrate the SIR parameters, we follow the Optimization based on Forward Model Simulation approach proposed by \cite{Balan2023_SIR}, which is detailed below. It requires the numbers of susceptibles and recovered individuals in the beginning of the episode, and --most importantly-- the daily counts of new infections throughout the epidemic.

From Section~\ref{sec:calib_firms_params}, the maximum firm size is $K := \max_{i\in\OneN} K_i = 12$. The cyber episode runs from 1 May to 31 July 2024; we denote the daily observation dates by $\{t_0, t_1,\hdots,t_T\}$. For calibration, we have from the dataset 
the daily detected infections, $(i_{k,t_u})_{1\leq k\leq K, 0\leq u\leq T}$. We reasonably assume that the initial removed are zero i.e. $(r_{k,t_0} = 0)_{1\leq k\leq K}$. As we do not have the initial number of susceptible population previously noted $(s_{k,t_0})_{1\leq k\leq K}$, we consider them as model parameters. \\
Assuming the firm size distribution follows a Zipf law with parameters $\af = 1.759$ and $q = 0.784$ (see \Cref{eq:pareto sub-unit}) for firm size in the insurance portfolio and the entire database, it is sufficient to consider the total number of firms $h_\star$ as a model parameter since for all $k=1\hdots K$, 
\begin{align}\label{eq:number firms per size}
    h_{k,t_0}= \left\lfloor h_\star\,\frac{\frac{1}{k^{\af}}}{\sum_{j=1}^K \frac{1}{j^{\af}}}\right\rfloor,
\end{align}
with $h_{k,t_0} = s_{k,t_0} + i_{k,t_0} + r_{k,t_0}$. Here,  $\lfloor x\rfloor$ denotes the floor function, representing the greatest integer less than or equal to $x$. \\
For the SIR model,  the main challenge consists in calibrating the coefficients $\varphi_0$,$\mu_\varphi$, $\kappa_\varphi$, and $\Sigma_\varphi$ of each stochastic SIR parameter $\varphi\in\Psi$, recalling from \eqref{eq:varphi solution} that, 
\begin{align*}
    \tilde\varphi_t = \varphi_0 e^{-\kappa_\varphi t} + \mu_\varphi\left( 1 - e^{-\kappa_\varphi t} \right) + \Sigma_\varphi \int_0^t e^{-\kappa_\varphi (t-s)} \sqrt{\tilde\varphi_s} \, d\cW_{\varphi,s}.
\end{align*}
In words, $\tilde\varphi$ is a CIR process where the speed of mean reversion is $\kappa_\varphi$, $\mu_\varphi$ the long-term mean level, $\Sigma_\varphi$ the volatility factor and, $\varphi_0$ is the initial condition. We have at this stage $4 \times (2K + K^2)$ CIR coefficients to calibrate, that is $672$ because $K = 12$. This is large relatively to the limited data at hand.
In order to simplify and speed up the calibration, we make the following additional assumptions.
\begin{enumerate}
    \item As stated in \Cref{rem:simplify b}, $(b_{j,k})$ is binomial with in-firm attack rate $a$ also governed by a CIR process.
    \item All stochastic SIR parameters, modeled as CIR processes, share the same mean reversion speed  and volatility coefficient. However, they possess different long-term mean levels, each set equal to its respective initial value, i.e. for all $\varphi\in\Psi$, $(\kappa_\varphi = \kappa_a)$, $(\Sigma_\varphi = \Sigma_a)$, and $(\mu_\varphi = \varphi_0)$.
    \item As state in \cite{doenges2024sir}, for each size, when the out-firm reproduction number ($\beta_k / \gamma_k$) is independent of the firm size and the infections are parallel, the recovery rate $\gamma_{k}$ is proportional to the inverse of the expected recovery time $\sum_{j=1}^{k}\frac{1}{j}$. We write
\begin{align}\label{eq:gamma_beta}
    \gamma_k = \frac{\gamma_1}{\sum_{j=1}^{k}\frac{1}{j}} \quad\text{and}\quad \beta_k = \frac{\beta_1}{\sum_{j=1}^{k}\frac{1}{j}},
\end{align}
where the last equality assumes that the decay rate of immunity is independent of the firm size. Therefore, it is enough to determine $\gamma_1$ and $\beta_1$. 
\end{enumerate}
We therefore need to calibrate six coefficients:  $\kappa_a$, $\Sigma_a$, $a_{t_0}$  $\gamma_{1,t_0}$, $\beta_{1,t_0}$, and $h_\star$, that are denoted as
\begin{align}
    \Theta := \left\{\kappa_a,\Sigma_a,a_{t_0},\gamma_{1,t_0},\beta_{1,t_0},h_\star\right\}.
\end{align}
The calibration procedure for $\Theta$ is as follows. A set $\Theta \in (\mathbb{R}_+)^6$ of CIR coefficients be given, we simulate $M \in \mathbb{N}^*$ trajectories of the SIR parameters over the discrete time grid $\{t_0, t_1, \dots, t_T\}$ using the exact simulation method for CIR processes described in \cite{alfonsi2015simulation}. For each parameter $\varphi \in \Psi$, we obtain the realizations $(\varphi_{t_u}^m)_{1 \leq m \leq M, 0 \leq u \leq T}$. Subsequently, for each trajectory $m\in\llbracket 1, M \rrbracket$, we evolve the SIR model using the following Euler scheme to generate $(S_k^m, I_k^m, R_k^m)_{1 \leq k \leq K}$. For $0 \leq u \leq T-1$ and $1 \leq k \leq K$, the dynamics are given by:
\begin{align}\label{eq:euler SIR}
\left\{
\begin{array}{l}
    \displaystyle S_{k,t_{u+1}}^{m} - S_{k,t_{u}}^{m} = -k Y_{t_{u}}^{m} S^{m}_{k,t_{u}} + Y_{t_{u}}^{m} \sum_{j=k+1}^{K} j S^{m}_{j,t_{u}}\cdot b^m_{j j-k,t_{u}},\\
    \displaystyle I^{m}_{k,t_{u+1}} - I_{k,t_{u}}^{m} = -\gamma^m_{k,t} I_{k,t_{u}} + Y_{t_{u}}^{m}\sum_{j=k}^{K} j S^{m}_{j,t_{u}}\cdot b^m_{j k, t_{u}} ,\\
    \displaystyle R^{m}_{k,t_{u+1}}-R^{m}_{k,t_{u}} = \gamma^m_{k,t_{u}} I^{m}_{k,t_{u}}, 
\end{array}
\right. \text{ where } Y^{m}_{t_{u}} = \frac{1}{N_0} \sum_{k=1}^{K} \beta^m_{k,t_{u}} \cdot k I^{m}_{k,t_{u}},
\end{align}    
with initial conditions $I_{k,t_{0}}^{m} = I_{k,t_{0}}$, $R_{k,t_{0}}^{m} = 0$, and $S_{k,t_{0}}^{m} = S_{k,t_{0}}$. The optimal coefficient $\Theta_\star$ minimizes the mean square error $J_2$, between the historical number of infected $(I_{k,t_u})$ plotted on \Cref{fig:Historical_infected} and the simulated trajectory $(I^{m}_{k,t_u})$, defined as
\begin{align}\label{eq:SIRS_Calibration}
    J_2(\Theta) := \frac{1}{M}\sum_{m=1}^{M}J^{m}_2(\Theta), \quad
\text{with}  \quad J^{m}_2(\Theta):= \sum_{k=1}^{K} \sum_{u=0}^{T} \left|I_{k,t_u} -I^{m}_{k,t_u} \right|^2. 
\end{align}
The procedure is  summarized in \Cref{alg:SIRS_Calibration}.
\begin{algorithm}[ht!]
\caption{SIR Calibration}
\label{alg:SIRS_Calibration}
\begin{algorithmic}[1]
\Procedure{SIR Calibration}{$K,T, (I_{k,t_u}), \af, q, M$}
    \State \textbf{Input:} Maximum firm size $K$, horizon of the cyber-episode $T$, number of infected $(i_{k,t_u})_{1\leq k\leq K, 0\leq u\leq T}$, Zipf law parameters $\af, q$.
    \ForAll{$\Theta = \left\{\kappa_a,\Sigma_a,a_{t_0},\gamma_{1,t_0},\beta_{1,t_0},h_\star\right\}$}
        \State Calculate the number of firms per size $(h_{k,t_0})$ using \eqref{eq:number firms per size}, then the initial number of susceptible using $(s_{k,t_0} = h_{k,t_0}-i_{k,t_0})$
        \State Convert numbers of susceptible into susceptible rates $S_{k,t_0} = \frac{s_{k,t_0}}{h_{k,t_0}}$, and numbers of infected into infection rates $I_{k,t_u} = \frac{i_{k,t_u}}{h_{k,t_0}}$.
        \State Calculate the average size $N_0 = \sum_{l=1}^{K} k(S_{k,t_0}+I_{k,t_0})$.
        \State Simulate $M$ trajectories $(\varphi_t^m)_{1\leq m\leq M,0\leq t\leq T}$ for each $\varphi\in\Psi$ on $\{0,1,\hdots,T\}$.
        \ForAll{$(\varphi_{t_u}^m)_{\varphi\in\Psi, 0\leq u\leq T}$}
            \State Simulate a forward SIR model with parameters $(\varphi_{t_u}^m)$ and initial condition $I_{k,t_0}^{m} = I_{k,t_0}$, $R_{k,t_0}^{m} = 0$, $I_{k,t_0}^{m} = I_{k,t_0}$, as well as $N_0$, and obtain daily time series $(S_k^{m}, I_k^{m}, R^{m}_k)_{1\leq k\leq K}$.
            \State Compute the objective function $J^{m}_2(\Theta)$.
        \EndFor
        \State Compute the objective function $J_2(\Theta)$.
    \EndFor
    \State Determine the minimum and the minimizer $\Theta$ of $J_2$.
    \State \textbf{Output:} the parameters $\Theta=(\kappa_a,\Sigma_a,a_{t_0},\gamma_{1,t_0},\beta_{1,t_0},h_\star)$.
\EndProcedure
\end{algorithmic}
\end{algorithm}
\Cref{tab:init_SIR_param} gives the 6 calibrated parameters. We observe that (1) the volatility of the stochastic SIR parameters is low, and (2) the initial in-firm infection rate, $\frac{1}{1+e^{-a_{t_0}}} = 0.586$ is relatively high.

    \begin{table}[ht!]
    \centering
    \begin{tabular}{|c|c|c|}\hline
     Initial recovery rate  & $\gamma_{1,t_0}$  & 0.6782\\ \hline
    Initial ”out–firm” infection rate & $\beta_{1,t_0}$ &0.5471\\ \hline
    Initial in-firm attack rate & $a_{t_0}$ & 0.3466\\ \hline
    Volatility & $\Sigma_\varphi$ & 0.0151\\ \hline
    Mean reversion speed  & $\mu_\varphi$ &0.4474 \\ \hline
    Total population of firms & $h_\star$ & 14,210\\ \hline
    \end{tabular}
    \caption{The SIR model parameters}
    \label{tab:init_SIR_param}
\end{table}
\noindent  The total number of firms $h_\star$ is used to calculate the initial population in each group, that we give in \Cref{tab:Hi}.
\begin{table*}[ht!]
    \centering
    \begin{tabular}{|r|r|r|r|r|r|r|r|r|r|r|r|r|}
\hline
& \textbf{$H_1$} & \textbf{$H_2$} & \textbf{$H_3$} & \textbf{$H_4$} & \textbf{$H_5$} & \textbf{$H_6$} & \textbf{$H_7$} & \textbf{$H_8$} & \textbf{$H_9$} & \textbf{$H_{10}$} & \textbf{$H_{11}$} & \textbf{$H_{12}$} \\
\hline
Size & 11144&  1646&   538&   244&   132&    80&    52&    36&    26&
           20&    15&    12\\
\hline
\end{tabular}
    \caption{The initial population per group}
    \label{tab:Hi}
\end{table*}

\subsection{Estimation procedure of revenues and claims, simulations, and discussion}

From the calibrated model coefficients $\Theta_\star$, we simulate  $M$ trajectories of the stochastic SIR parameters $(\varphi^m)_{\varphi\in\Psi,1\leq m\leq M}$ ($M= 10,000$ in this numerical analysis). For each trajectory $m$ of parameters, we simulate the cyber-contagion using the Euler scheme for the SIR described in \eqref{eq:euler SIR}. We obtain, for each scenario $1\leq m\leq M$, the force of epidemics $(Y_{t_u}^{m})$. 

\paragraph{Simulation of the arrival of cyber-attack on each firm (and on each subunit)} We fix the scenario $1\leq m\leq M$. 
For each firm $1\leq i\leq H$, we simulate $(\tau_{ij}^{m})_{1\leq j\leq K_i,0\leq u\leq T}$ as follows:
\begin{enumerate}
    \item Primary (or out-firm) infections: for each subunit $1\leq j\leq K_i$, we simulate the first jump of the Cox process $(N_{ij,t_u}^{0,m,sim})_{0\leq u\leq T}$ with intensity $(Y_{t_u}^{m})_{0\leq u\leq T}$. 
    \item Determine the time of the very first cyber-attack in the firm $i$, $\tau_i^{m}$ using \eqref{eq:tau_i} with the convention $\tau_i^{m} = T+1$ if $\displaystyle\sum_{j=1}^{K_i} N_{ij,T}^{0,m,sim} = 0$. This means that firm $i$ is not suffering from a primary infection.
    
    If $K_i = 1$, then $\tau_{i,j}^{m} = \tau_i^{m}$.
    \item If $K_i \geq 2$ and $\tau_i^{m} \leq T$ (that is firm $i$ has many subunits and at least one is externally infected), then secondary (or in-firm) infections can occur: for each $1\leq j\leq K_i$, if $N_{ij,\tau_i^{m}}^{0,m,sim} = 0$ (that is subunit $j$ is not yet infected), simulate a  random variable $U_{ij}^{m}$ with success parameter $a_{\tau_i^{m}}^{m}$
    \begin{itemize}
     \item if $U_{ij}^{m} = 1$, set $\tau_{i,j}^{m} = \tau_i^{m}$,
     \item else, set $\tau_{i,j}^{m} = T+1$ (subunit $j$ is not suffering from either a primary infection or a secondary infection).
    \end{itemize}
    \item If $K_i \geq 2$ and $\tau_i^{m} = T+1$, $\tau_{i,j}^{m} = +\infty$ for all $1\leq j\leq K_i$.
\end{enumerate}

\paragraph{Simulation of the costs}
We simulate also  $M$ realizations of $(\pi_{ij}^{m})_{1\leq m\leq M}$ of the random variables $\pi_{ij}$ for $(i,j)\in\Jj$, to  compute the Monte Carlo approximation of the total instantaneous expected revenue at each date of the cyber episode using the last equality of \eqref{eq:whole economy}: for each $0\leq u\leq T$
\begin{equation}
    \EE[\mathbf{O}_{t_u}] \approx  \sum_{i=1}^{H} \sum_{j=1}^{K_i} z_{ij,0} e^{\mu_{ij} u} \left(1- \frac{1}{M}\sum_{m=1}^{M} \pi_{ij}^{m}\bOne_{ \tau_{i,j}^{m} \leq t_u < \tau_{i,j}^{m} + \frac{1}{\gamma^{m}_{i,\tau_{i,j}^{m}}}} \right).
\end{equation}
For all $1\leq i\leq H$, the daily claims of firm $i$ for each scenario $m$ and at each date $t_u$, $\CC_{i,[t_u,t_{u+1}]}^{m}$, is given, using \eqref{eq:claim in period}, by the trapezoidal method between $t_u\vee\tau_{i,j}^{m}$ and $t_{u+1}\wedge\left(\tau_{i,j}^{m} + \frac{1}{\gamma^{m}_{i,\tau_{i,j}^{m}}}\right)$. We then approach the cumulative distribution function (cdf) of the portfolio at each date $t_u$, with $u\in\{0,1,\hdots,T\}$ 
\begin{align*}
     \ff_{[t_u,t_{u+1}]}(x) &= \PP\left(\sum_{i=1}^{H}C_{i,[t_u,t_{u+1}]} \leq x\right), \quad x\geq 0, 
\end{align*}
by the empirical cdf 
\begin{equation}\label{eq:empirical CDF}
\widehat\ff_{[t_u,t_{u+1}]}^M(x) := \frac{1}{M} \sum_{m=1}^{M} \bOne_{\left\{\displaystyle \sum_{i=1}^{H}C_{i,[t_u,t_{u+1}]}^{m} \leq x \right\}},
\end{equation}
where $\bOne_{(X > x)}$ is the indicator function that equals $1$ if $X > x$,
 and $0$ otherwise. \Cref{alg:compute_cdf} details the procedure for the  compute of $\ff_{[u,u+1]}^M(x)$.
\begin{algorithm}[h!]
    \caption{Calculate CDF $\ff$}
\label{alg:compute_cdf}
\begin{algorithmic}
    \Procedure{$\ff$}{$x, T, M,\alpha^\pi, \beta^\pi, \Theta, (S_{k,t_0}, I_{k,t_0}, R_{k,t_0})$}.
    \State \textbf{Input:} quantile $x$, horizon $T$, number of simulations $M$, the set of coefficients of the SIR parameters $\Theta$, the initial values of the SIR $(S_{k,t_0}, I_{k,t_0}, R_{k,t_0})_{1\leq k\leq K}$, parameters of severity $(\alpha^\pi, \beta^\pi)$, and the parameters of firms $(z_{ij,t_0},\mu_{ij},\sigma_{ij})_{(i,j)\in\Jj}$.
    \ForAll{Scenario $1\leq m\leq M$}
    \State From $\Theta$, simulate a trajectory of the set of SIR parameters $\Psi$.  \Comment{Line 1}
    \State Simulate a SIR model with parameter $\Psi$ and initial condition $(S_{k,t_0}, I_{k,t_0},R_{k,t_0})_{1\leq k\leq K}$.  
    \State Simulate the infection times. 
    \State Simulate a trajectory of severity with $(\alpha^\pi, \beta^\pi)$. 
    \ForAll{Date $1\leq u\leq T$}
    \State Calculate the set of infected firms at time $u$.
    \State Calculate the total claims at $t_u$ using $\sum_{i=1}^{H}C_{i,[t_u,t_{u+1}]}^m$. \Comment{Line 6}
    \EndFor
    \State Calculate $\ff$ at $u$ from \eqref{eq:empirical CDF}.  
    \EndFor
    \State \textbf{Output:} The cumulative distribution function of daily losses $\ff$.
    \EndProcedure
\end{algorithmic}
\end{algorithm}
\noindent Finally, we compute the  Aggregate Exceedance Probability (AEP) over $[0,T]$ defined in \eqref{def:portfolio measure}  by performing  Monte Carlo simulations described in~\Cref{alg:compute_aep} with $M_P = 10,000$. 
\begin{algorithm}[h!]
    \caption{Calculate $\aep$}
\label{alg:compute_aep}
\begin{algorithmic}
    \Procedure{$\aep$}{$x,\upsilon, M_P$ and parameters of $\ff$ in \Cref{alg:compute_cdf}}
    \State \textbf{Input:} AEP threshold $x$, number of simulations of AEP $M_P$, intensity of Poisson random variable $\upsilon$.
    \State Simulate $M_P$ times $P\sim \text{Poisson}(\upsilon)$ and get $(P_m)_{m\in\llbracket 1,M_P\rrbracket}$.
    \ForAll{$P_m$}
        \State Simulate $P_m$ the loss $\CC_T$ by repeatedly executing \textit{Line 1} to \textit{Line 6} of \Cref{alg:compute_cdf} on the interval $[t_0, t_T]$ instead of $[t_u, t_{u+1}]$ and get $(\CC_T^{l,m})_{1\leq l\leq P_m}$.
    \EndFor
    \State Calculate $\displaystyle \aep_T(x) = \displaystyle\frac{1}{M_P} \sum_{m=1}^{M_P} \bOne_{\left\{\displaystyle \sum_{l=1}^{P_m} \CC_{T}^{l,m} > x \right\}}$.
    \State \textbf{Output:} The aggregate exceedance probability $\aep_T$ at $T$.
    \EndProcedure
\end{algorithmic}
\end{algorithm}

\subsubsection{Results of the contagion model}
To compare the realized dynamics of the cyber‑contagion (see \Cref{fig:Historical_infected}) with the simulated ones, we focus on the daily number of infected subunits, the peak defined in \eqref{eq:peak}, and the date at which it occurs.  \Cref{fig:Peak_with_confident_interval} illustrates the Lockbit trajectory (black) alongside the 99\% confidence interval (blue) for the simulated dynamics over $T = 100$ days, based on $M = 10,000$ trajectories. Initial inspection of the results indicates that the model aligns closely with the observed epidemic trajectory: both the observed and mean simulated trajectories begin from the same initial condition (32 infected firms  equivalent to  49 infected subunits) and reach their peak -- defined as the maximum of the averaged trajectory -- on day 38. While the peak date is identical, our model overestimates the observed peak by 17.15\%.
\begin{figure}[ht!]
    \centering
    \includegraphics[width=0.75\linewidth]{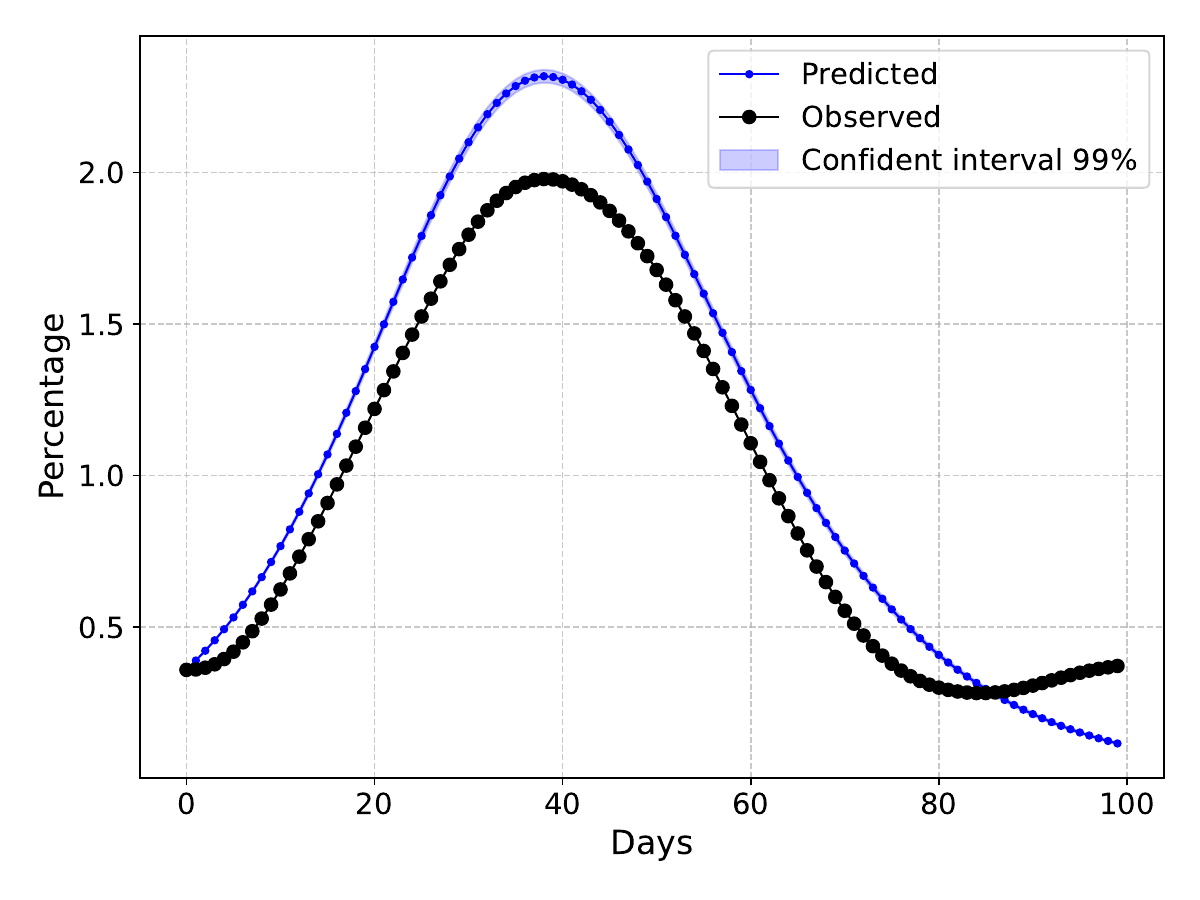}
    \caption{The dynamics of the total number of infected subunits, as percentage of the total population $h_\star$ }
    \label{fig:Peak_with_confident_interval}
\end{figure}
As in \cite{hillairet2021propagation}, the peak (and the date it is reached) can be used to model the limited capacity of the insurance company to respond to an incident. If the number of policyholders needing help is too large, the insurer's assistance teams can become overloaded. \Cref{fig:Peak_date__histogram} displays the distribution of peak dates across all simulated trajectories. The predicted peak is in close alignment with the realized peak: it predicts 317 infected subunits (to be compared to  273 observed), representing a 16.12\% overestimation. Notably, while the observed peak occurs on day 38, the model accurately predicts the same peak date, with a maximum absolute deviation of 12 days across all simulations.
\begin{figure}[ht!]
    \centering
    \begin{subfigure}[b]{0.45\textwidth}
        \centering
        \includegraphics[width=\textwidth]{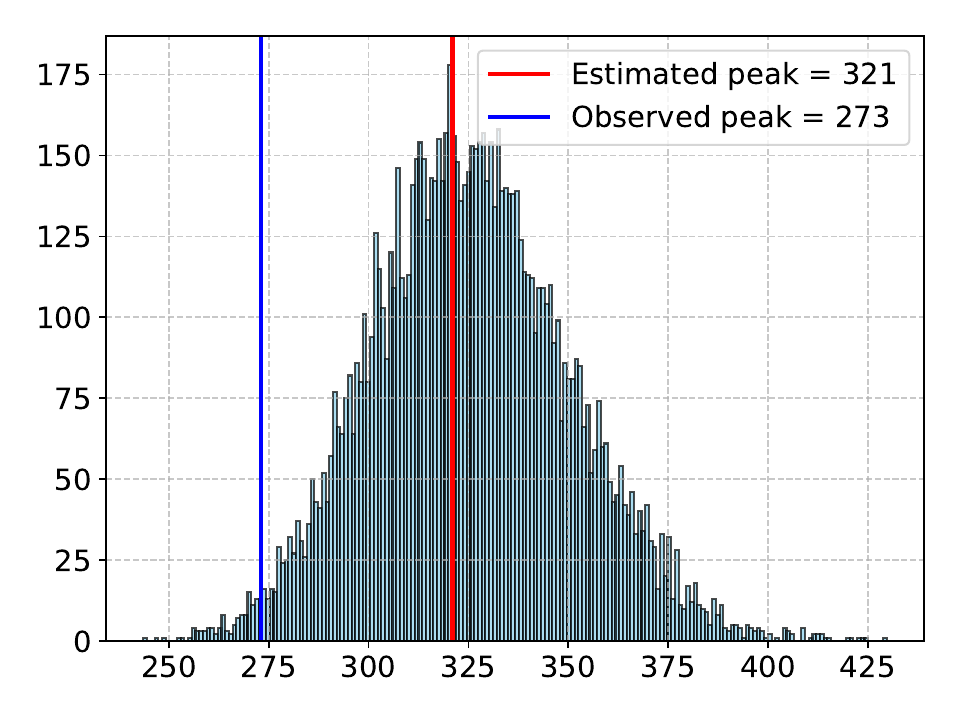}
        \caption{The value of the peak}
        \label{fig:Peak_height_histogram}
    \end{subfigure}
    \hfill
    \begin{subfigure}[b]{0.45\textwidth}
        \centering
        \includegraphics[width=\textwidth]{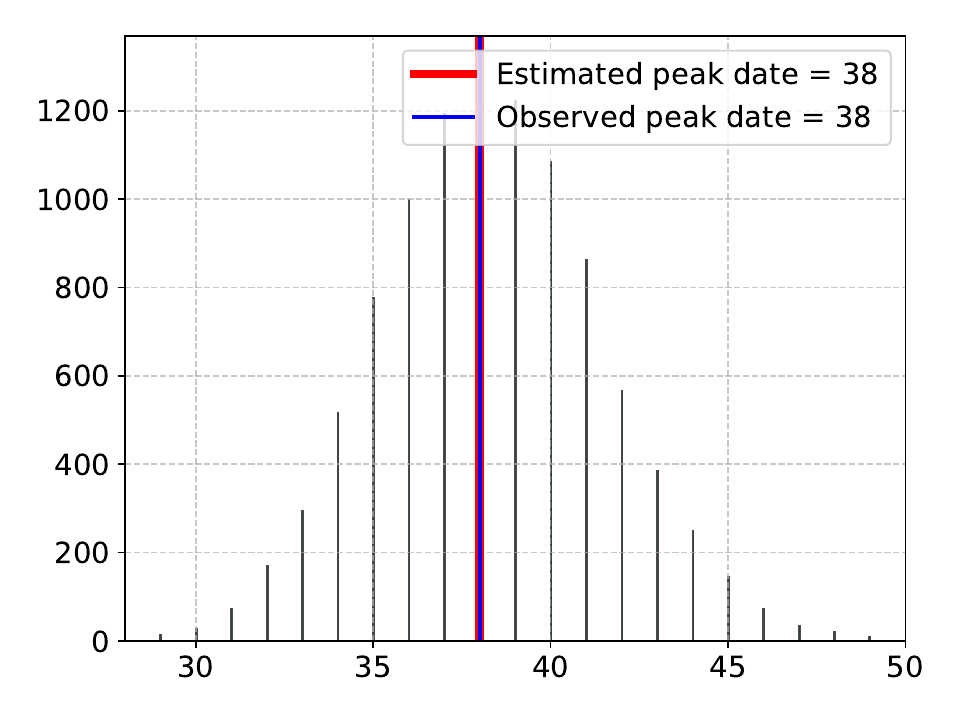}
        \caption{The date at which the peak is reached}
        \label{fig:Peak_date__histogram}
    \end{subfigure}
    \caption{Histograms}
    \label{fig:peak_height}
\end{figure}
In fact, it is preferable to slightly overestimate both the peak and underestimate the date at which it is reached, thereby yielding conservative risk measures.\\
We now examine the dynamics of contagion by firms size. \Cref{fig:SIR} represents the evolution of 
\begin{figure}[ht!]
    \centering
    \begin{subfigure}[b]{0.325\textwidth}
        \centering
        \includegraphics[width=\textwidth]{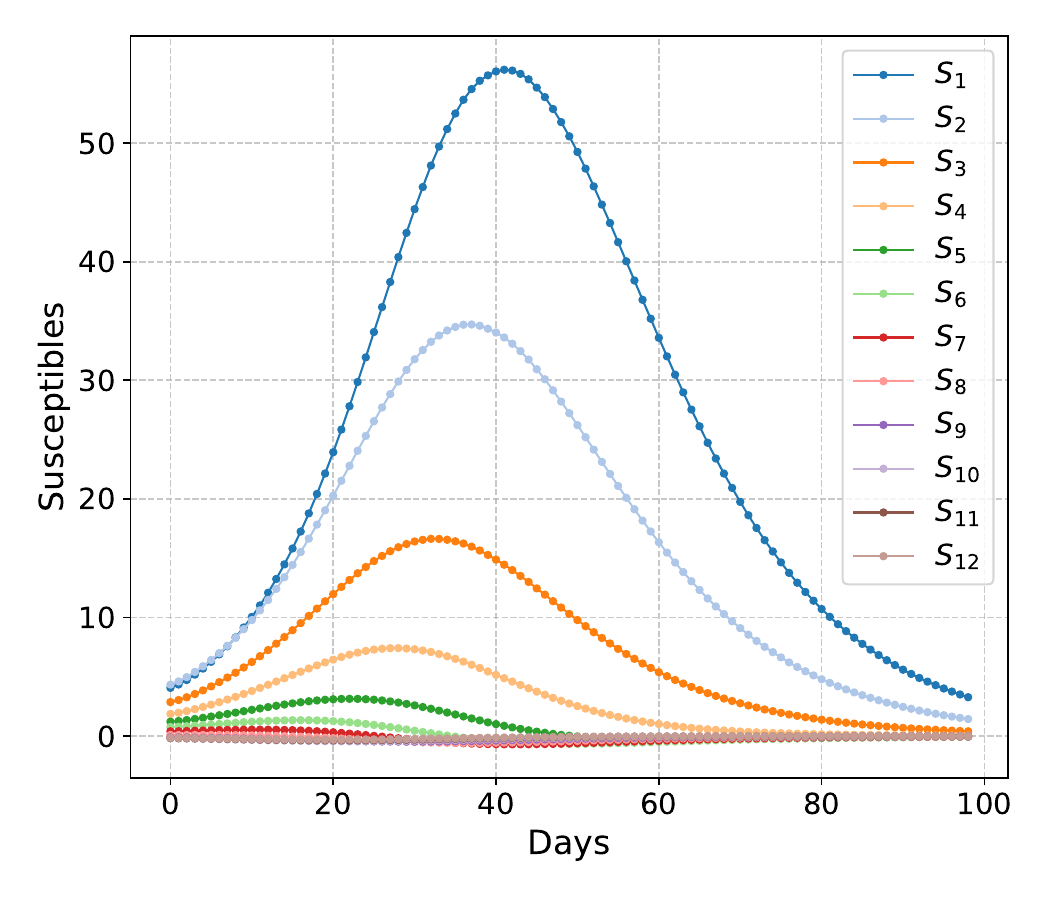}
        \caption{New susceptible}
        \label{fig:Susceptible}
    \end{subfigure}
    \hfill
    \begin{subfigure}[b]{0.325\textwidth}
        \centering
        \includegraphics[width=\textwidth]{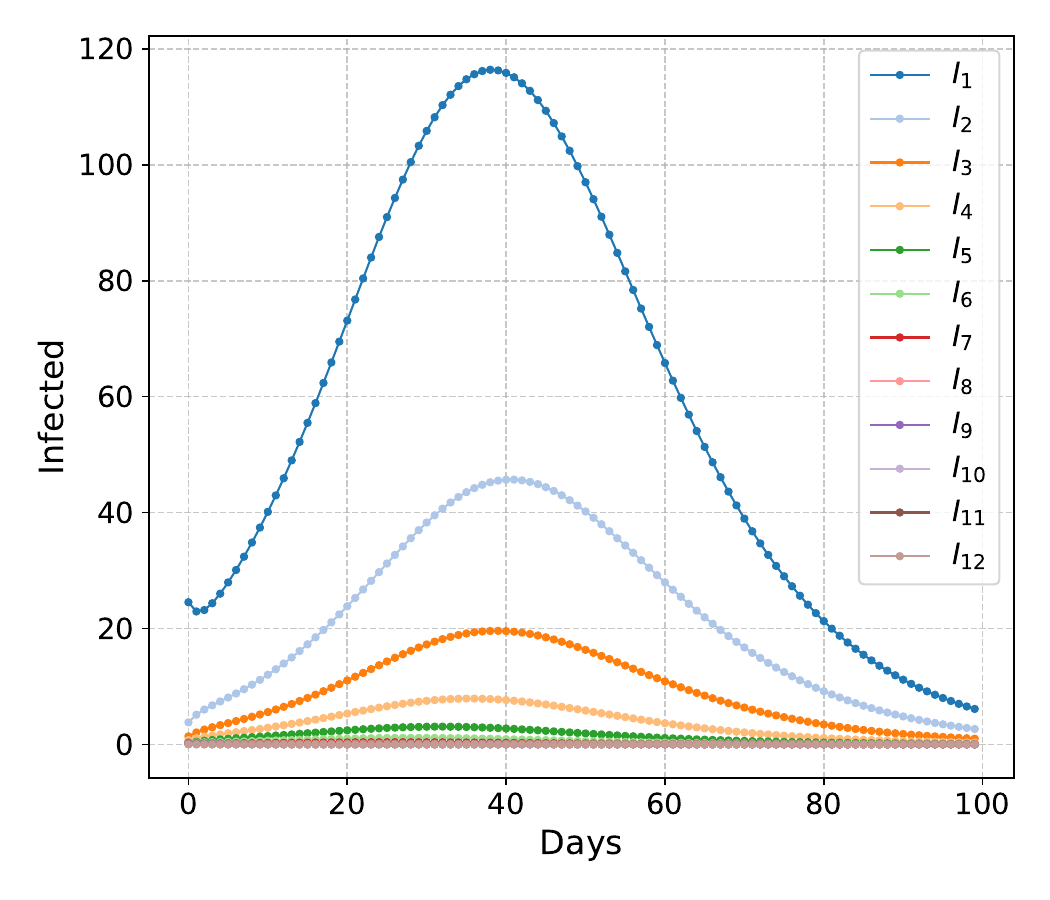}
        \caption{Infected}
        \label{fig:Infected}
    \end{subfigure}
    \hfill
    \begin{subfigure}[b]{0.325\textwidth}
        \centering
        \includegraphics[width=\textwidth]{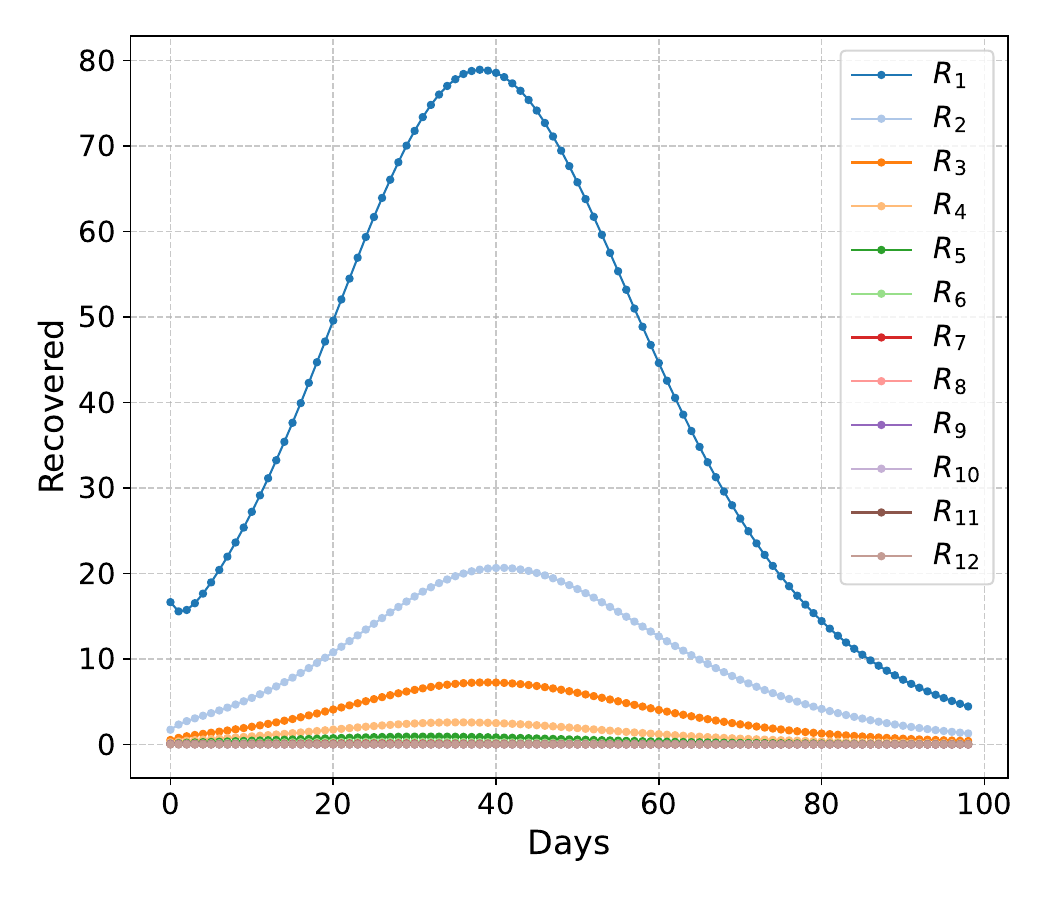}
        \caption{New recovered}
        \label{fig:Recovered}
    \end{subfigure}
    \caption{Evolution of the population in  each group as percentage of the initial population in the group (see \Cref{tab:Hi}) }
    \label{fig:SIR}
\end{figure}
susceptible (\Cref{fig:Susceptible}), infected (\Cref{fig:Infected}), and recovered (\Cref{fig:Recovered}) firms by size. The three dynamics follow a similar pattern -- characterized by rapid growth toward a peak followed by a steady decline. However, different sizes reach their respective peaks at different dates. The peak dates by firm size are summarized in \Cref{tab:peak_height_date_per_group}. 
\begin{center}
    \begin{table*}[ht]
    \centering
    \begin{tabular}{r|r|r|r|r|r|r|r|r|r|r|r|r}
$k = $& \textbf{1} & \textbf{2} & \textbf{3} & \textbf{4} & \textbf{5} & \textbf{6} & \textbf{7} & \textbf{8-12} \\
\hline
New susceptibles height date & 41 &37& 33& 28& 22& 16 &11 & 0\\
\hline
Infected height date &38& 40& 39& 35& 32& 28& 26&  0\\
\hline
    Recovered height date & 38& 40& 39& 35& 32& 28& 26&  0\\
\end{tabular}
    \caption{The date of the peak's height per size}
    \label{tab:peak_height_date_per_group}
\end{table*}
\end{center}
For each group, the epidemic dynamics exhibit a consistent temporal ordering. For firms of size 1, the peak is closest or equal to the peak of epidemics (i.e. the maximum of the total number of infected subunits for all the size). This is consistent with the fact that this group represents 78\% of companies. For firms of size 2 to 8, the peak in the number of infected and recovered individuals occur at the same time. This is consistent because the recovery rate is particularly high ($\gamma_{1,t_0} = 0.678$ in \Cref{tab:init_SIR_param}) which, according to \cite{doenges2024sir}, is proportional to the inverse of the expected recovery time. The peak in the susceptible population for each size  occurs before the peak of the total infected population. The peak decreases as the size increases, indicating a  collapse of the epidemics in the larger organizational units first. This is partly due to the structure of the model, in which large firms feed into small ones. Moreover, since firm's size follows Zipf's law, there are much less firms with big size than small size. In our example, there are exactly 109 firms from size 8 to 12 which is $0.77\%$ of the population (see \Cref{tab:Hi}). Thus, large companies very quickly reach their peak (right from the start of the epidemic here as shown in the last column of \Cref{tab:peak_height_date_per_group}).
Additional analyzes of the contagion model -- especially regarding parameter sensitivity -- can be conducted, drawing inspiration from  \cite{doenges2024sir}. 

\begin{figure}[ht!]
        \centering
        \includegraphics[width=0.55\textwidth]{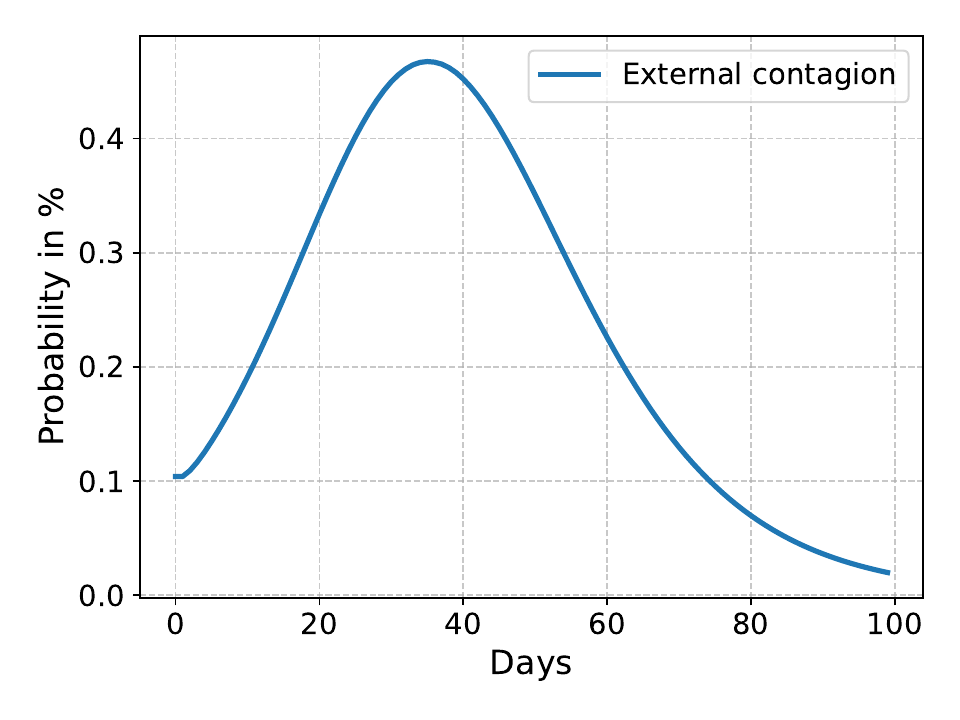}
        \caption{Probability of external infection}
        \label{fig:Y_average_trajectory}
\end{figure}
We now turn to the probabilities of internal and external contagion. \Cref{fig:Y_average_trajectory} shows, for each day of the epidemic, the probability that the infection of a subunit originates outside the firm. This probability is low, though not zero and mirrors the epidemic dynamics, with a peak around day 38. By contrast, the probability that the infection of a subunit comes from another subunit within the same firm that we denoted $\frac{1}{1+e^{-a}}$,
is high ($>58\%$ calculated using the parameters \Cref{tab:init_SIR_param}) and, as assumed, independent of the progression of the epidemic. In summary, the likelihood that a firm becomes infected is small; however, once infection occurs, there is a significant chance that multiple subunits will be affected.

\subsubsection{Impact of the contagion on a given firm}

This section aims to examine the broader impact of contagion on the economy (see \eqref{eq:whole economy}) or its representative firms (i.e. the arithmetic mean of firms per  size $1\leq k\leq K$, $\overline{Z}_{t}^{(k)} := \frac{\sum_{i=1}^{H} \mathbf{1}_{\{K_i=k\}}\,Z_{i,t}}{\sum_{i=1}^{H}\mathbf{1}_{\{K_i=k\}}}$). It is important to recall that contagion effects stem from the organizational structure of the firms within the portfolio -- specifically, the distribution of firm sizes and the presence of subsidiary networks. While current data constraints preclude the estimation of firm-specific or sectoral cost differences, the model is designed with the flexibility to incorporate these factors should the data become available.

As introduced in \Cref{ass:sub-unit dynamics}, when a subunit undergoes a cyber-attack, its revenue negatively jumps by $\pi\sim\mathcal{B}(\alpha^\pi, \beta^\pi)$. For the following simulations, we assume $\alpha^\pi = 50$ and $\beta^\pi = 10$ giving a mean loss of $0.833$ with a standard deviation of $0.048$. 
We start by calculating $K_i$ for each firm $i$ and $z_{0,ij}$, $\mu_{ij}$, $\sigma_{ij}$ for each subunit $j$ of firm $i$. Then we run $M = 10,000$ scenarios corresponding to the simulation of the first arrival times of the cyber-event $\tau_{ij}$ and of the severity $\pi_{ij}$. 

We plot on \Cref{fig:First_infection_time} the probability density function of the infection time for different firm size i.e. $(\tau_k)_{1\leq k\leq 12}$. Since we focus on the first 100 days of the epidemic, the event $\{\tau_k \geq 100\}$ -- whose probability is reported in \Cref{tab:First_infection_time}-- corresponds to ‘no infection’ within the first 100 days.. 

\begin{center}
    \begin{table*}[ht]
    \centering\footnotesize
    \begin{tabular}{r|r|r|r|r|r|r|r|r|r|r|r|r}
$k = $& \textbf{1} & \textbf{2} & \textbf{3} & \textbf{4} & \textbf{5} & \textbf{6} & \textbf{7} & \textbf{8} & \textbf{9} & \textbf{10} & \textbf{11} & \textbf{12} \\
\hline
    $\PP[\tau_k> 100]$ & 0.794&0.623&0.500&0.394&0.310& 0.249& 0.194& 0.150& 0.129& 0.095& 0.069& 0.060\\
\end{tabular}
    \caption{Probability of no infection ($\tau_k > 100$)}
    \label{tab:First_infection_time}
\end{table*}

\begin{figure}[ht!]
        \centering
        \includegraphics[width=0.75\textwidth]{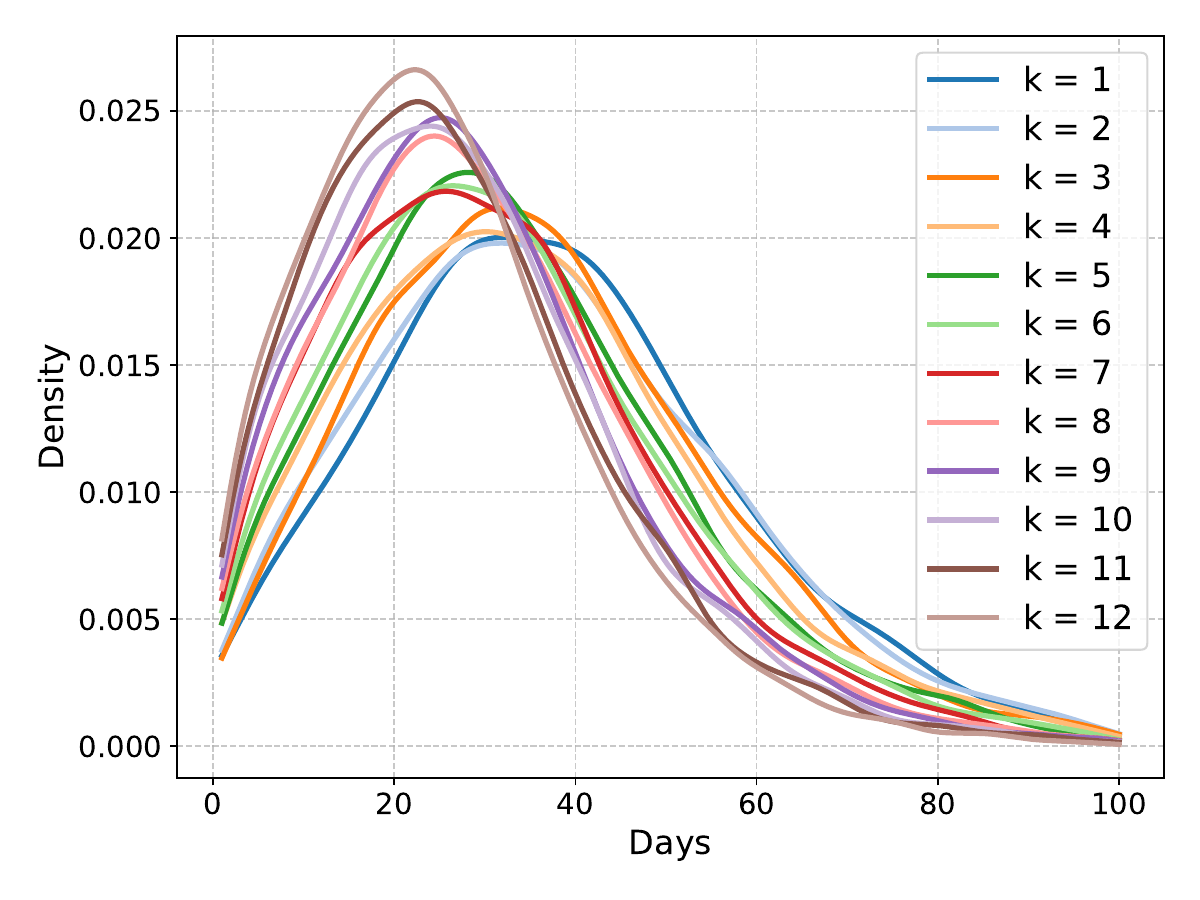}
        \caption{PDF of the infection time}
        \label{fig:First_infection_time}
\end{figure}
\end{center}

From \Cref{tab:First_infection_time}, we observe that the larger the firm, the lower the probability that it becomes infected. In \Cref{fig:First_infection_time},  the distribution mode -- the point at which the probability density reaches its maximum -- is reached well before (on day 24) the peak of the epidemic (which only occurs on day 38). This mode naturally decreases with the size of the company. This means that the larger the size, the greater the chance that the firm will be infected before the peak is reached. Therefore, our model reproduces the stylized fact of cyber-episodes mentioned in \cite{BaksyCyberFirmSize2025, kamiya2021risk}.

For each group $k=1, \hdots, 12$, we plot on \Cref{fig:Firm_revenue_loss} the average daily firm's claim as a percentage of the daily revenue, calculated across $M = 10,000$ simulations. 
\begin{figure}[ht!]
    \centering
    \begin{subfigure}[b]{0.48\textwidth}
        \centering
        \includegraphics[width=\columnwidth]{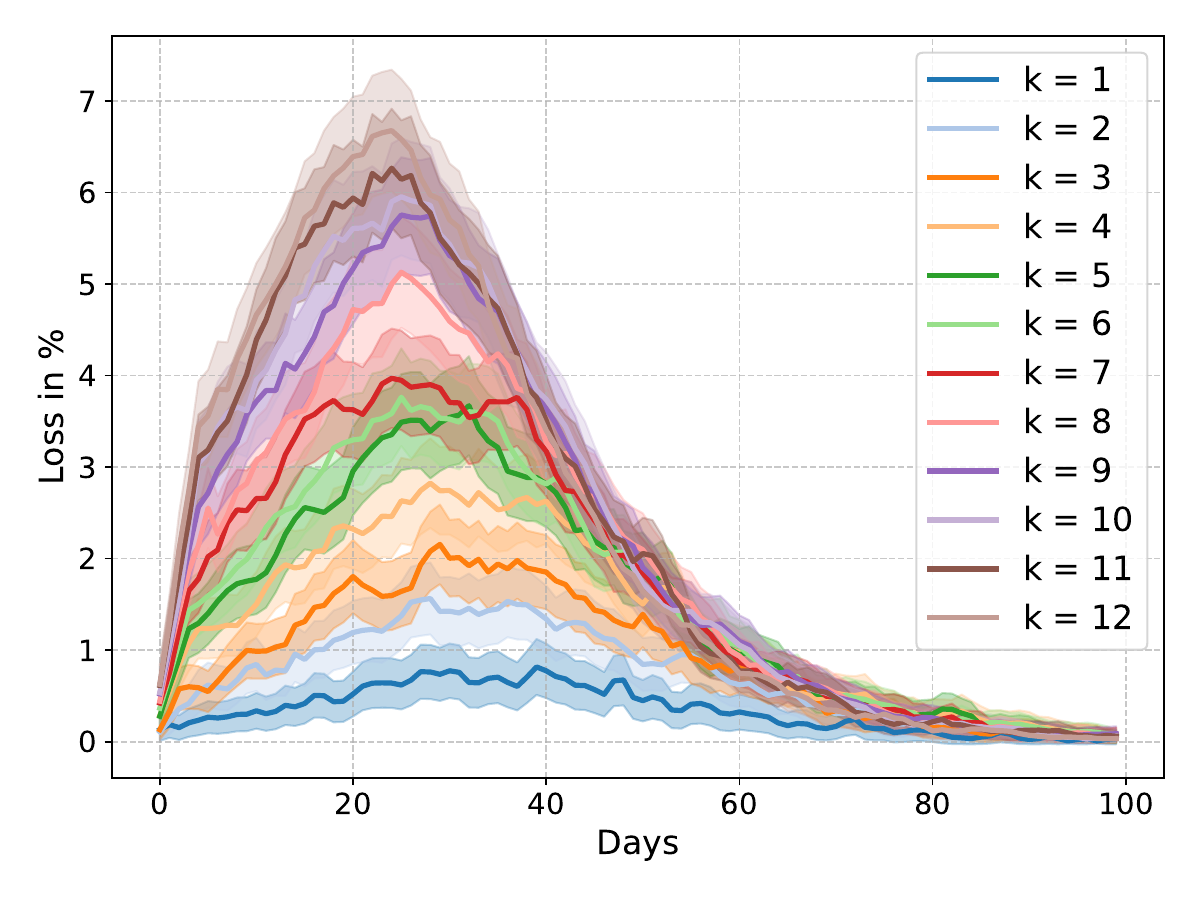}
        \caption{Mean across $10,000$ simulations}
        \label{fig:Firm_revenue_loss}
    \end{subfigure}
    \hfill
    \begin{subfigure}[b]{0.48\textwidth}
        \centering
        \includegraphics[width=\columnwidth]{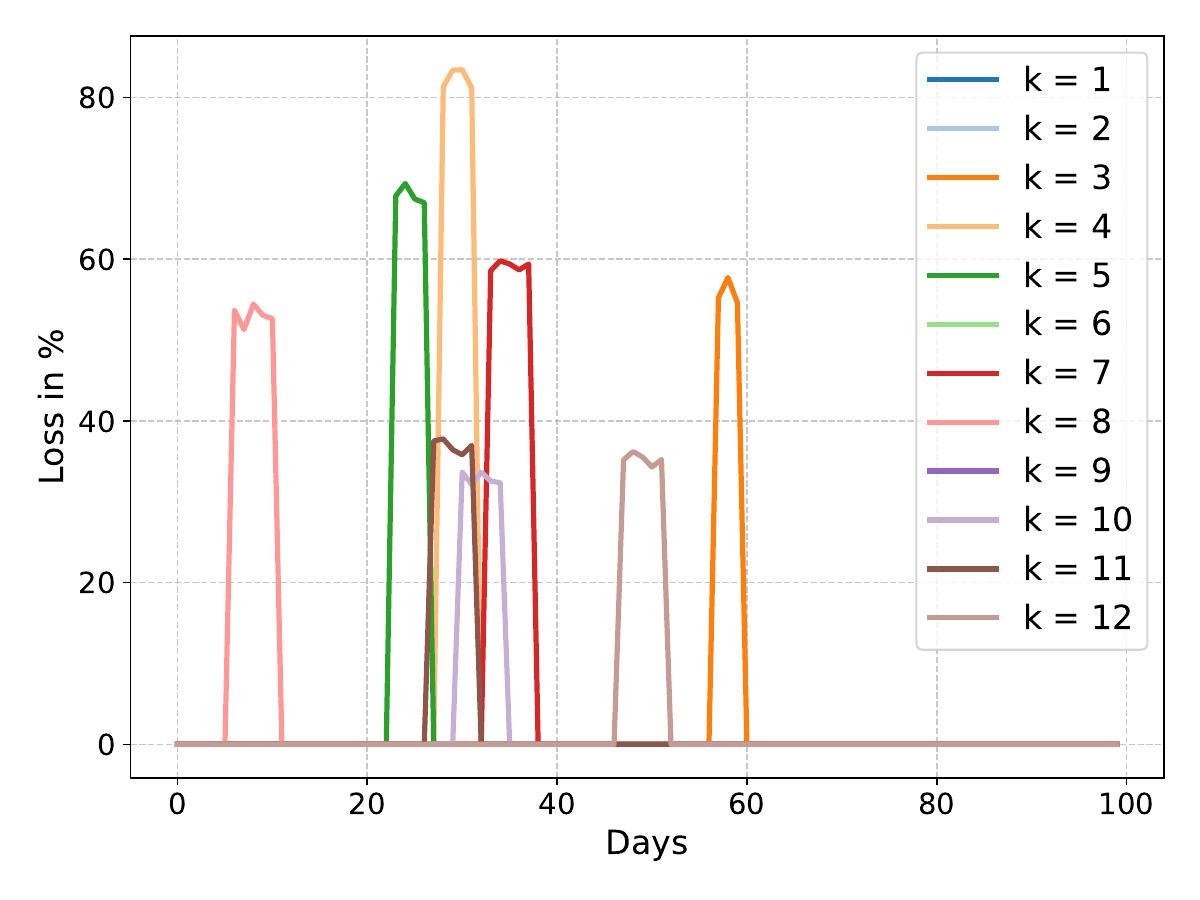}
        \caption{Single simulation trajectory}
    \label{fig:losses_1attack}
    \end{subfigure}
    \caption{Daily revenue loss by firm size}
    \label{fig:losses}
\end{figure}
As expected, for all sizes, the potential loss of firm is naturally correlated with the date of initial infection. Furthermore, the losses increase with the number of subunits, clearly implying that size is an aggravating factor. This is due to internal contamination between subsidiaries whose probability is  58.6\%. 

While an average daily loss peaking at $7\%$ may appear low, this figure masks the underlying volatility. If we move beyond the mean of $10,000$ simulations and consider a single contagion scenario, the picture changes significantly. As illustrated in \Cref{fig:losses_1attack}, even with sporadic attacks, the claims relative to daily revenue can be substantial, occasionally reaching 75\%. This extreme variance underscores the critical need for insurance. Furthermore, the duration of the infection increases with company size, as reflected by the recovery period $\delta_{ij} = \frac{1}{\gamma_{i,\tau_{ij}}}$ defined in \eqref{eq:delta def}.

\subsubsection{Impact of the contagion on an insurance portfolio}

We now proceed to the computation of the Aggregate Exceedance Probability (AEP) outlined in \Cref{alg:compute_aep}. We begin by analyzing the PDF and the tails of the portfolio daily losses using their CDF computed in \Cref{alg:compute_cdf}. Given that small businesses typically hold fewer insurance policies, we focus on companies with a size greater than or equal to 2 ($K_i \geq 2$). This subset comprises 621 firms whose total revenue on the first day of the cyber-episode is \euro 39.13 million.

\paragraph{Probability density function (PDF) of portfolio daily losses}

We simulate $M = 10,000$ scenarios of the epidemic which gives $M$ trajectories of the daily total losses $\sum_{i=1}^{621} C_{i,[t_u,t_{u+1}]}$ defined in \eqref{eq:claim in period}. We analyze results on 5 different dates: $t_u=7$ (7 days after the beginning of the epidemic), $t_u =24$ (14 days before the peak), $t_u=38$ (the peak date), $t_u=52$ (14 days after the peak), and $t_u=93$ (7 days before the end of the epidemic).\\
\begin{figure}[ht!]
    \centering
    \begin{subfigure}[b]{0.48\textwidth}
        \centering
    \includegraphics[width=\textwidth]{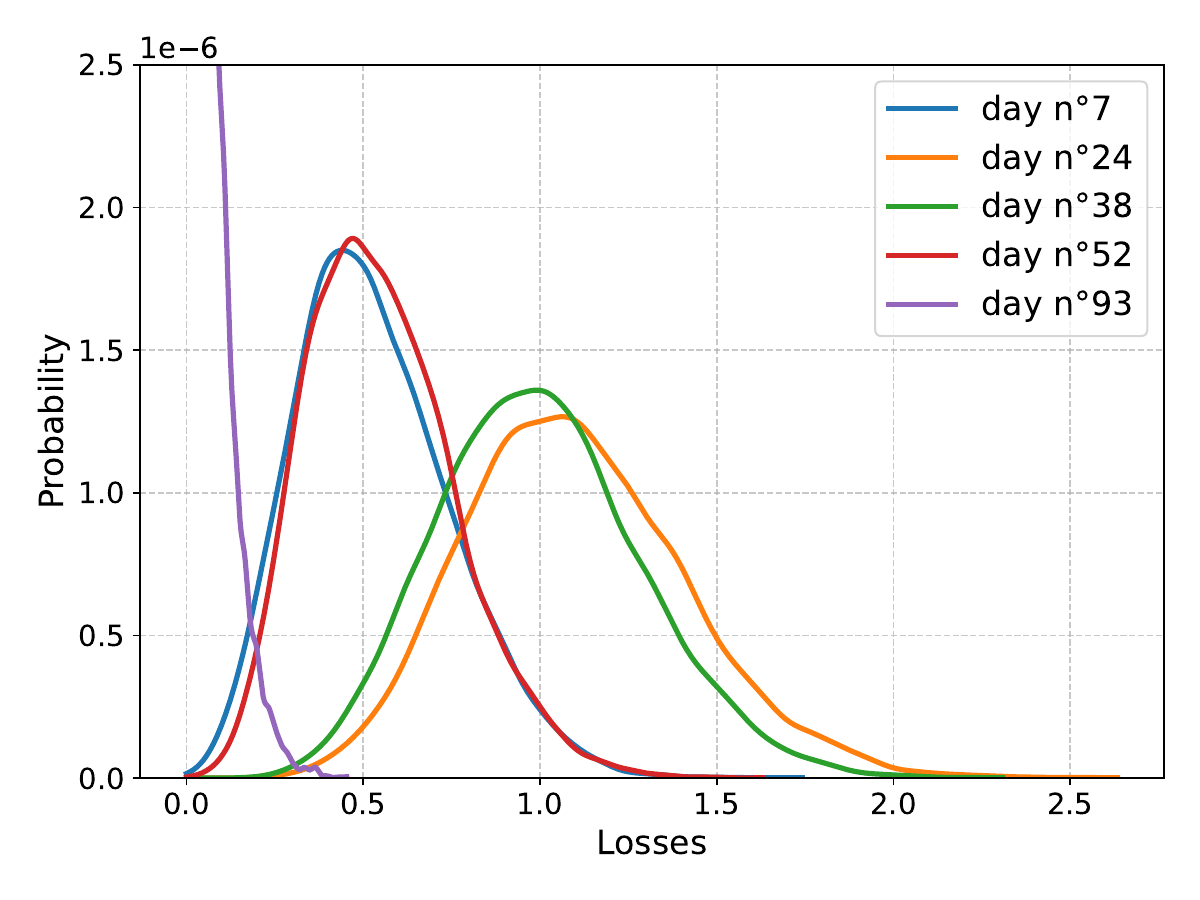}
    \caption{PDF curve}
\label{fig:PDF_costs_whole_portfolio}
    \end{subfigure}
    \hfill
    \begin{subfigure}[b]{0.48\textwidth}
        \centering
    \includegraphics[width=\columnwidth]{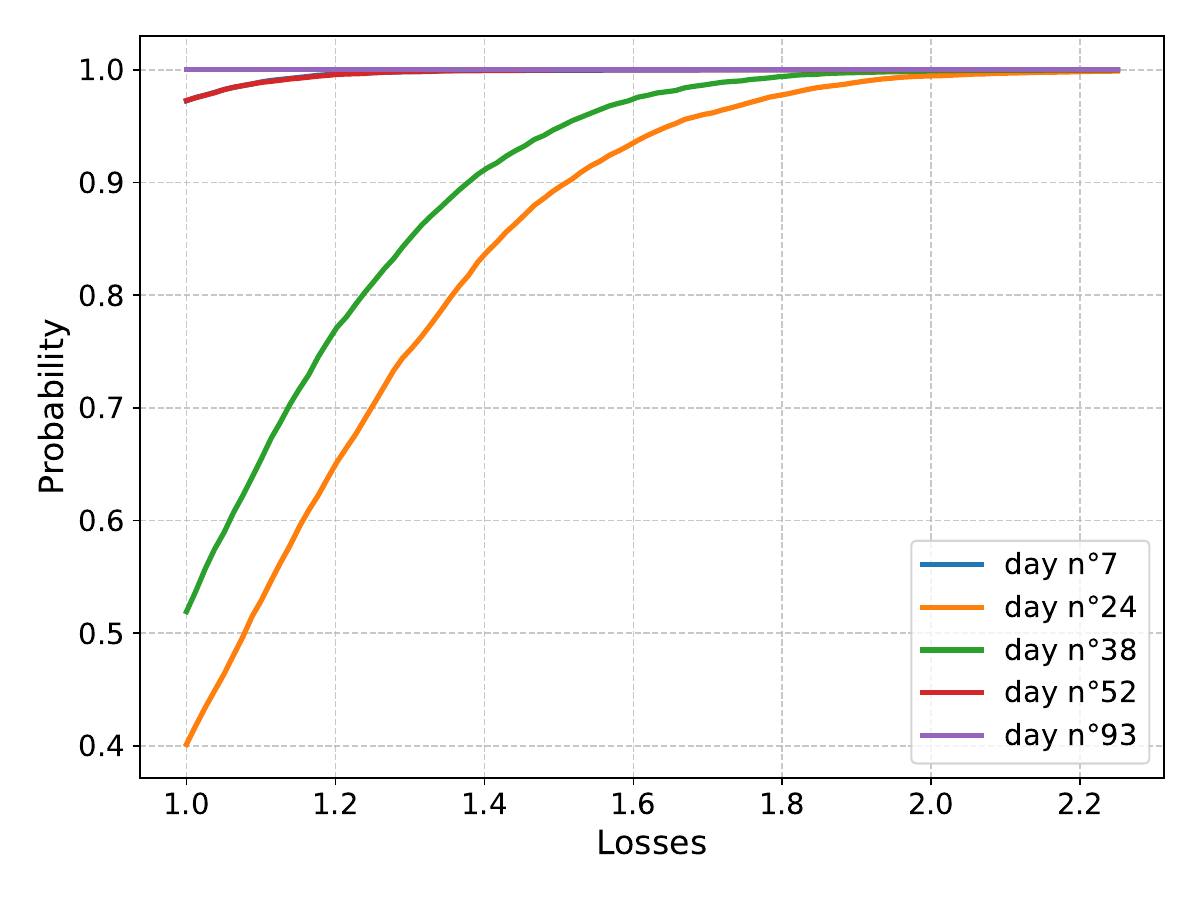}
    \caption{Right tails of the CDF}
    \label{fig:CDF_costs_whole_portfolio}
    \end{subfigure}
    \caption{Daily financial impact in \euro million}
    \label{fig:pdf_and_cdf}
\end{figure}
\Cref{fig:PDF_costs_whole_portfolio} and \Cref{tab:Characteristics_distribution_support} provides the PDF of portfolio daily losses. We observe that the closer we are to the peak of the cyber-event, the wider and flatter the total claims distribution is. This increased dispersion is most pronounced around day 24, coinciding with the period of highest infection probability -- the mode of the infection-time distribution -- as shown in \Cref{fig:First_infection_time}. The distribution’s support widens near the peak of the cyber‑contagion, and both its mean and mode are higher around that point than on other days. It can also be noted that the mode and the mean of the distribution at $t_u = 7$ are larger than at $t_u=93$, reflecting a rapid decline. To summarize, the cyber-event induces a deformation in the distribution of total losses: as the contagion nears its peak (roughly between days 24 and 52), the PDF of daily losses becomes increasingly flattened, reflecting a broader dispersion around the mode.

\begin{table}[ht]
    \centering
    \begin{tabular}{c|c|c|c|c|c}
        \textbf{Day} & \textbf{7} & \textbf{24} & \textbf{38} & \textbf{52} & \textbf{93}\\ \hline
        \textbf{Lower bound} & 0  & 0  & 0 & 0 & 0
 \\ \hline
        \textbf{Upper bound} & 1.742& 2.634& 2.309& 1.630&  0.453\\ \hline
         \textbf{Mode} & 0.441& 1.062& 0.991& 0.470& 0.001\\ \hline
    \textbf{Mean} &0.871& 1.317&1.154& 0.815& 0.227 \\
    \end{tabular}
    \caption{Characteristics of the distribution's support (in \euro million )}

\label{tab:Characteristics_distribution_support}
\end{table}
The losses on x-axis are in million euros and  should be compared to the average daily revenue of a company in the portfolio, i.e. \euro 39.13 million. In other words, we could expect a daily claim of up to 6.729\% of \euro 39.13 million. Let zoom on the right tail of the distribution.

\paragraph{Right tails of distribution of the portfolio daily losses}
 We consider claims from \euro 1 million (2.56\% of the average daily revenue of the portfolio) to \euro 2.634 million (6.729\%) using the empirical CDF given in \eqref{eq:empirical CDF}. \Cref{fig:CDF_costs_whole_portfolio} illustrates the right tails of the portfolio daily losses at 5 different dates. Consistent with the PDF deformation on \Cref{fig:PDF_costs_whole_portfolio}, the right tail of distribution of daily losses thickens with the severity of the cyber-event. On day 24 when the mode of first arrival time is reached, the tail is the heaviest and the probability that the portfolio daily losses exceed \euro 1 million reaches 60\% whereas it remains negligible around the initial and terminal phases of the epidemic.\\
\noindent Similar as in \Cref{fig:pdf_and_cdf},   \Cref{fig:PDF_cum_costs_whole_portfolio} and \Cref{fig:CDF_cum_costs_whole_portfolio}
provides the  PDF and the right tails of the total losses incurred by the sub-portfolio of 621 firms (with size $K_i$ greater or equal 2) $\CC_T = \sum_{i=1}^{621} C_{i,[t_0,t_{100}]}$, over the 100-day duration of the cyber-episode.
\begin{figure}[ht!]
    \centering
    \begin{subfigure}[b]{0.485\textwidth}
        \centering
        \includegraphics[width=\textwidth]{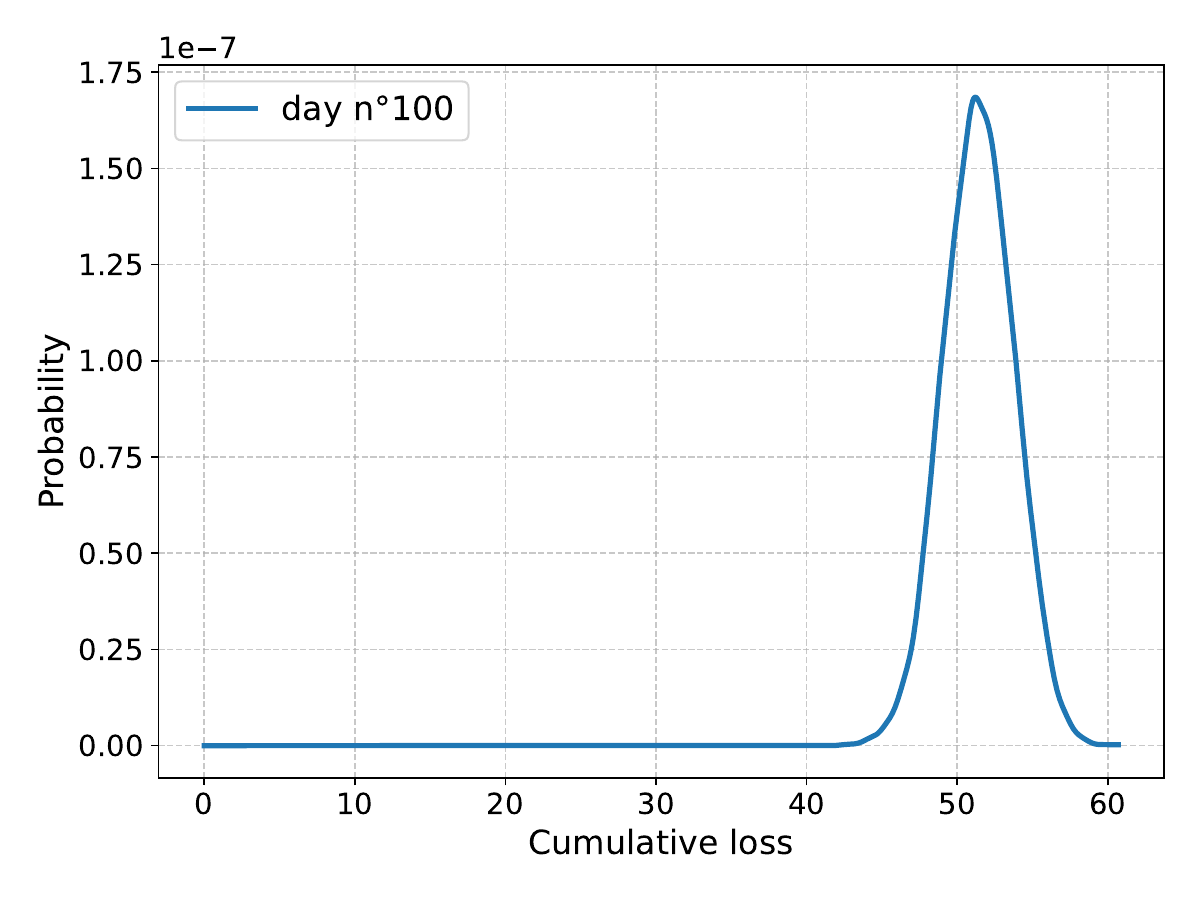}
        \caption{PDF curve}
        \label{fig:PDF_cum_costs_whole_portfolio}
    \end{subfigure}
    \hfill
    \begin{subfigure}[b]{0.485\textwidth}
        \centering
        \includegraphics[width=\textwidth]{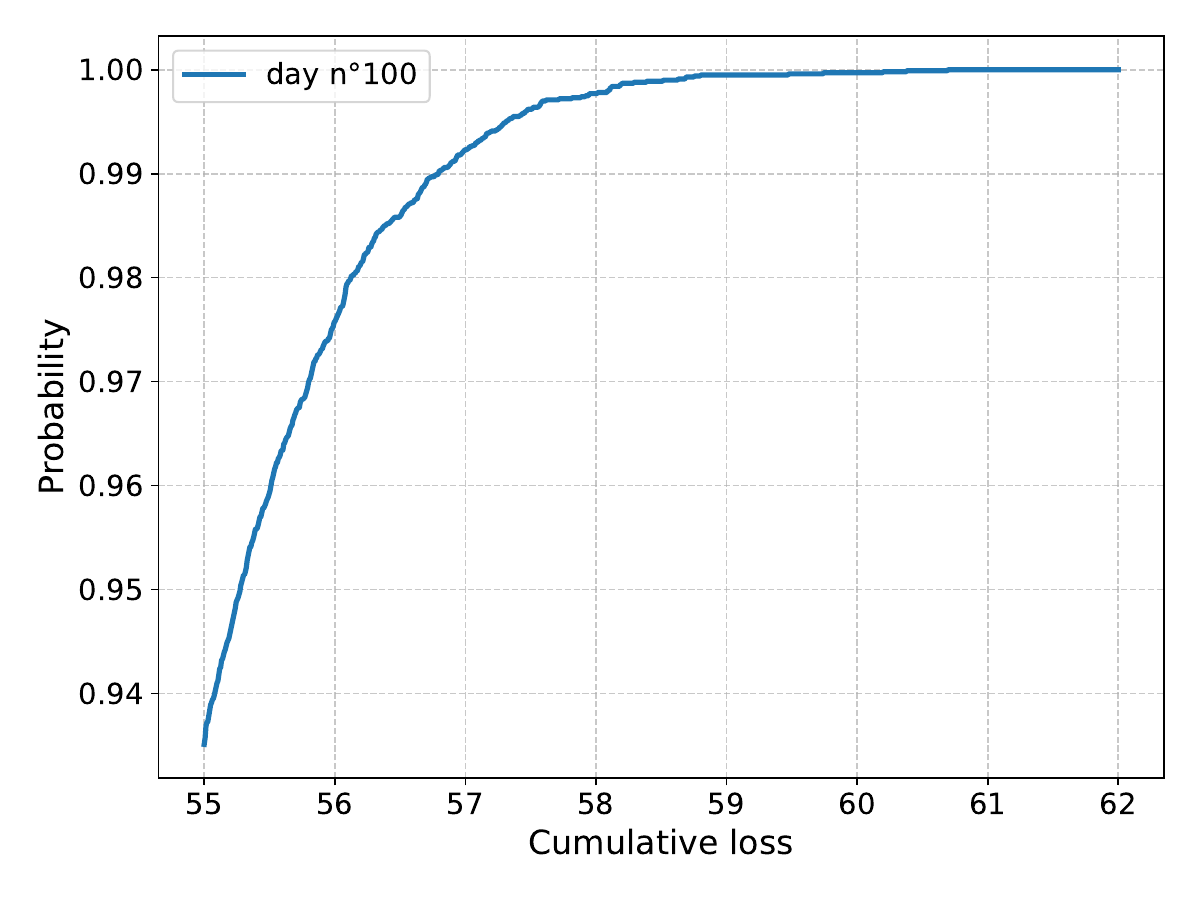}
        \caption{Right tail of the CDF}
        \label{fig:CDF_cum_costs_whole_portfolio}
    \end{subfigure}
    \caption{Financial impact after 100 days of the epidemic}
\end{figure}
The PDF curve in \Cref{fig:PDF_cum_costs_whole_portfolio} suggests that the total claims after 100-day may range approximately up to \euro 60.44 million, with a median greater than \euro 51.47 million and a mode equals to \euro 51.43 million. Moreover, the right tail in \Cref{fig:CDF_cum_costs_whole_portfolio} indicates that the total losses after 100 days exceeds \euro 55 million with a probability around 6\%. In practical terms, given a total daily revenue of \euro 39.13 million, there is a risk that the total claims from a single 100-day cyber-episode will exceed 1.5 days of aggregate revenue for all firms in the sub-portfolio.

 \paragraph{Aggregate Exceedance Probability (AEP) of the whole portfolio} \Cref{fig:AEP_curve} displays the exact AEP curve (in blue), computed via $\CC_T$ in \eqref{eq:c portfolio}, alongside the approximated AEP curve (in red), computed via $\CC_T^\star$ in \eqref{eq:approx C_T}. These results are obtained using \Cref{alg:compute_aep} with $M_P = 10,000$ and an average of $\upsilon = 0.105$ cyber-episodes per $T=100$ days. This means that there is a 10\% probability that a contagious cyber-episode occurs within 100 days (see \Cref{tab:poisson}).
\begin{table}[h!]
\centering
\begin{tabular}{c|cccc}
\hline
 & $P=0$ & $P=1$ & $P=2$ & $P\ge 3$ \\
\hline
Probability (\%) & 90.0 & 9.45 & 0.496 & 0.054 \\
\hline
\end{tabular}
\caption{ Poisson distribution with parameter $\upsilon = 0.105$.}
\label{tab:poisson}
\end{table}

\begin{figure}[ht!]
    \centering
    \begin{subfigure}[b]{0.48\textwidth}
        \centering
        \includegraphics[width=\columnwidth]{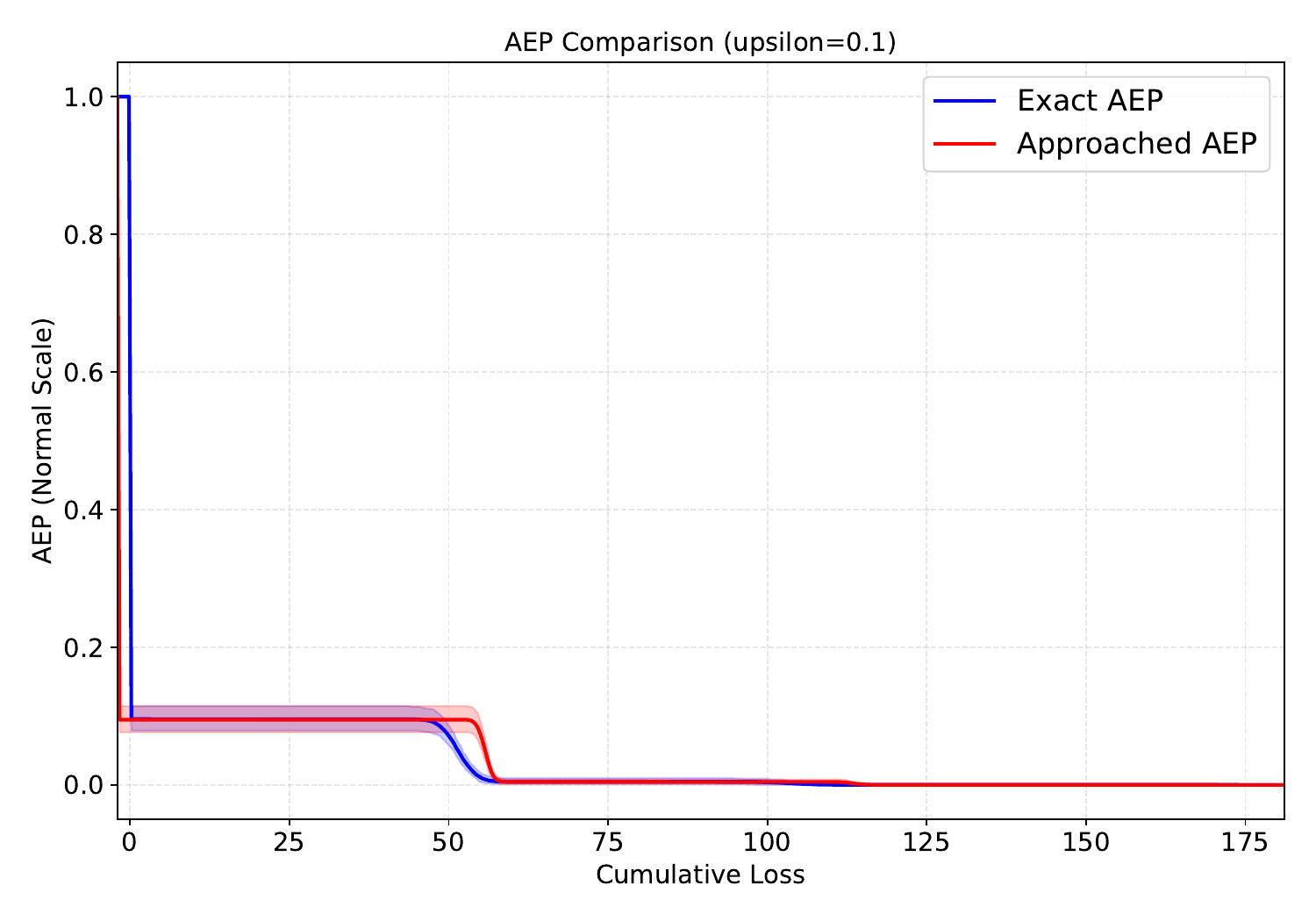}
        \caption{Normal scale}
        \label{fig:AEP_cum_costs_whole_portfolio}
    \end{subfigure}
    \hfill
    \begin{subfigure}[b]{0.48\textwidth}
        \centering
        \includegraphics[width=\columnwidth]{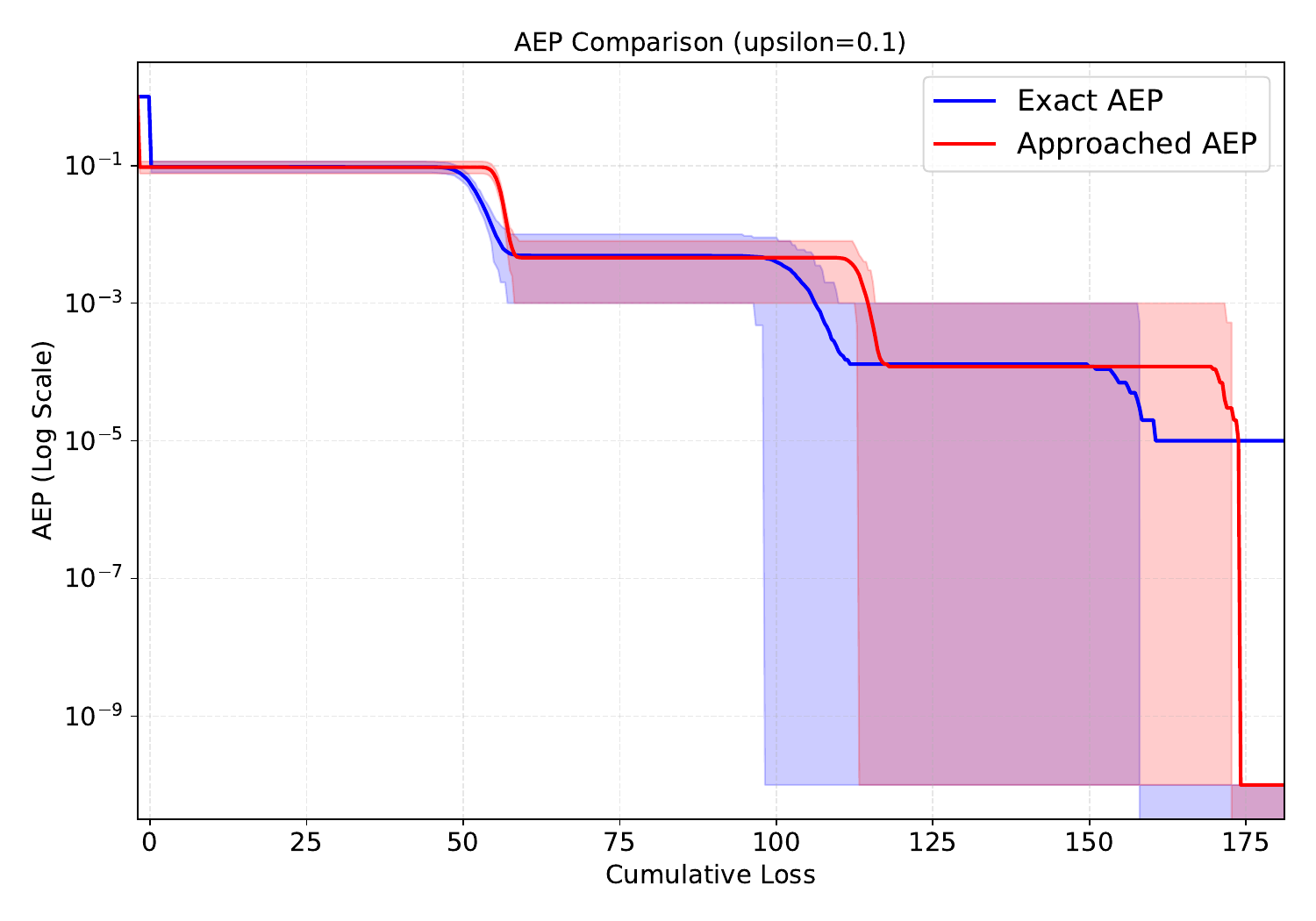}
        \caption{Log-scale}
        \label{fig:AEP_cum_whole_portfolio_log}
    \end{subfigure}
    \caption{Exact (in blue) and approached AEP (in red) curves (the sum of losses are in \euro{} million)} 
    \label{fig:AEP_curve}
\end{figure}
The exact AEP curve 
on \Cref{fig:AEP_cum_costs_whole_portfolio} shows that, since the probability that no event occurs is about $90\%$, most of the mass of the distribution is zero. Moreover, even with an average of only $0.1$ cyber-episode over $100$ days, distribution support appears to extend to \euro180 million -- roughly four days of revenue for the entire portfolio. We subsequently adopt a log-scale to better visualize this tail and to analyze the behavior conditional on the occurrence of cyber-episodes. {The exact EAP curve} on Figure~\ref{fig:AEP_cum_whole_portfolio_log} illustrates that, when an event occurs, total losses range from $0$ to \euro180 million with a decreasing, sometimes very small but strictly positive probability. As expected from the definition of $\aep$, the curve is composed of successive segments: the part between $0$ and \euro60 million corresponds to $F_{\CC_{100}}$; the section between roughly \euro60 and \euro120 million reflects the 2-fold convolution of $F_{\CC_{100}}$; and the tail between \euro120 and \euro180 million corresponds to the 3-fold convolution. {Additionally}, the initial flat portion of the curve stems directly from the shape of the single-episode loss distribution: conditional on an event, the probability that total losses fall below \euro40 million is zero (see \Cref{fig:PDF_cum_costs_whole_portfolio}). Finally, the approached AEP (the red curve in \Cref{fig:AEP_curve}, computed from $\CC_T^\star$ defined in \eqref{eq:approx C_T}) provides a good approximation of the exact AEP: it remains in the same range (up to \euro180 million) and exhibits the same jump structure. However, the jumps are noticeably steeper, and the approximation (that does not take into account the idiosyncratic risk) slightly underestimates the probability of the most extreme events. 
In our example, the portfolio is relatively small ($H=621$
 firms), so the AEP approximation is less relevant and the computational gain is negligible. For large-scale portfolios, by contrast, idiosyncratic risk mutualization operates more effectively and averaging out that risk becomes meaningful. The resulting computational efficiency is then significant enough that the slight underestimation of the AEP becomes an acceptable trade-off.   \\

The analysis can be extended  by calculating the AEP and PDF by activity's sector, by company size, and according to many other factors. We may also adopt different frequencies to examine how losses evolve on a weekly or monthly basis. Ultimately, the appropriate frequency and level of aggregation depend on the specific structure of the insurance portfolio.

\section*{Conclusion}
 This paper proposes a stochastic multigroup SIR model coupled with a granular model of firm growth, incorporating stylized facts from economics and cybersecurity to describe the propagation of cyber‑episodes within a cyber-insurance portfolio. We provide quantitative analysis of the impact of such an event  on firms’ revenue dynamics and on an insurance portfolio. The model shows that the infection of a subunit most likely originates from another subunit within the same firm, not from outside. Moreover, larger firms experience more internal transmission, leading to greater losses and higher insurance claims. This extends the previous works of \cite{hillairet2021propagation,hillairet2022cyber} and provides a more comprehensive framework to assess  the impact of a cyber-event on an economy or an insurance portfolio, under multiple scenario configurations.


\small
\bibliographystyle{apalike}
\bibliography{apssamp}

\end{document}